\documentclass[11pt]{article}
\usepackage[T1]{fontenc}
\usepackage{amsfonts}
\usepackage{amsmath}
\usepackage{amssymb}
\usepackage{amsthm}
\usepackage{bbm}
\usepackage{bm}
\usepackage{mathrsfs}
\usepackage{verbatim}
\usepackage{setspace}
\usepackage{color}
\usepackage{pdfsync}
\usepackage{enumitem}
\usepackage{graphicx}
\usepackage{tikz}
\usetikzlibrary{patterns}
\usepackage[margin=31 mm]{geometry}

\usepackage{tocloft}

\usepackage[pdfborder={0 0 0}]{hyperref}
\hypersetup{
  colorlinks=true,
  linkcolor=blue,
  citecolor=blue,
  urlcolor=blue,
  pdfauthor={Marcel Nutz, Moritz Voss},
  pdfkeywords={Stochastic optimal tracking, convergence rate, optimal execution, transient price impact},
  pdftitle={The Convergence Rate of Stochastic Tracking with Application to Optimal Execution},
  pdfsubject={Stochastic optimal tracking and optimal execution with transient price impact},
  pdfpagemode=UseNone
}

\usepackage[capitalize,noabbrev]{cleveref} %
\usepackage[showonlyrefs=true]{mathtools}

\theoremstyle{plain}
\newtheorem{theorem}{Theorem}[section]
\newtheorem{proposition}[theorem]{Proposition}
\newtheorem{lemma}[theorem]{Lemma}
\newtheorem{corollary}[theorem]{Corollary}
\theoremstyle{definition}

\newtheorem{remark}[theorem]{Remark}
\newtheorem{example}[theorem]{Example}
\newtheorem{assumption}[theorem]{Assumption}
\crefname{assumption}{Assumption}{Assumptions}
\Crefname{assumption}{Assumption}{Assumptions}
\theoremstyle{remark}
\AddToHook{env/theorem/begin}{\crefalias{section}{theorem}}
\AddToHook{env/proposition/begin}{\crefalias{theorem}{proposition}}
\AddToHook{env/lemma/begin}{\crefalias{theorem}{lemma}}
\AddToHook{env/corollary/begin}{\crefalias{theorem}{corollary}}
\AddToHook{env/definition/begin}{\crefalias{theorem}{definition}}
\AddToHook{env/remark/begin}{\crefalias{theorem}{remark}}
\AddToHook{env/example/begin}{\crefalias{theorem}{example}}
\AddToHook{env/assumption/begin}{\crefalias{theorem}{assumption}}
\crefname{theorem}{Theorem}{Theorems}
\crefname{proposition}{Proposition}{Propositions}
\crefname{lemma}{Lemma}{Lemmas}
\crefname{corollary}{Corollary}{Corollaries}
\crefname{definition}{Definition}{Definitions}
\crefname{remark}{Remark}{Remarks}
\crefname{example}{Example}{Examples}
\crefname{assumption}{Assumption}{Assumptions}

\renewenvironment{thebibliography}[1]{%
\begin{oldthebibliography}{#1}%
\linespread{1}
\small
\setlength{\parskip}{.35ex}%
\setlength{\itemsep}{.35em}%
}%
{%
\end{oldthebibliography}%
}
\newcommand{\eps}{\varepsilon}

\newcommand{\E}{\mathbb{E}}

\newcommand{\R}{\mathbb{R}}

\newcommand{\cF}{\mathcal{F}}

\newcommand{\cH}{\mathcal{H}}
\newcommand{\cY}{\mathcal{Y}}

\newcommand{\cQ}{\mathcal{Q}}
\newcommand{\cR}{\mathcal{R}}
\newcommand{\cS}{\mathcal{S}}

\newcommand{\cJ}{\mathcal{J}}

\newcommand{\sU}{\mathscr{U}}

\newcommand{\1}{\mathbf{1}}

\newcommand{\mysnorm}{[\xi]^2_{\rm{tr}}}

\newcommand{\mykill}[1]{}
\numberwithin{equation}{section}

\begin{document}

\title{\vspace{-0em}
  The Convergence Rate of Stochastic Tracking\\with Application to Optimal Execution
}
\author{Marcel Nutz\footnote{Departments of Statistics and Mathematics, Columbia University, mnutz@columbia.edu. Research supported by NSF Grants DMS-2106056, DMS-2407074, DMS-2606731.} \hspace{3ex} Moritz Voss\footnote{Department of Mathematics, University of California Los Angeles, voss@math.ucla.edu.}}
\date{\today}

\maketitle

\begin{abstract}
We study the quadratic tracking problem of a general stochastic target process with absolutely continuous controls, with and without terminal constraint. We derive explicit, non-asymptotic upper bounds in terms of a Besov-type modulus of the target. These bounds yield sharp explicit rates that specialize to the square-root order for semimartingale targets. We then apply these results to a generalized Obizhaeva--Wang execution model with random terminal inventory. We first develop a Hilbert-space approach to characterize its optimal strategy, which includes jumps. To avoid such trading spikes, one regularizes the problem by a quadratic trading-rate penalty with coefficient~$\varepsilon$. We then show that the regularized optimal execution cost---and therefore the excess price impact cost of the regularized optimal strategy---converges at the sharp rate \(O(\sqrt{\varepsilon})\). Since the regularized optimal strategy is not available in closed form, we further construct a nearly optimal strategy which is readily implementable and shares the same approximation rate.
\end{abstract}

\begin{description}
{\small \item[Mathematics Subject Classification (2020):] 93E20, 91G80, 60H30.    
\item[JEL Classification:] G11, C61.	

\item[Keywords:] Stochastic optimal tracking, convergence rate, optimal execution, transient price impact.}
\end{description}

\setcounter{tocdepth}{1}
\tableofcontents

\section{Introduction}

We consider stochastic optimal tracking problems of the form
\begin{equation}\label{def:tracking-intro}
 \underset{u}{\operatorname{min}}\;
 \E\left[
 \frac12\int_0^T (X_t^u-\xi_t)^2\,dt
 +
 \frac{\kappa}{2}\int_0^T u_t^2\,dt
 \right],
\end{equation}
in which a controller chooses the velocity
\(u=(u_t)_{0\le t\le T}\) of an absolutely continuous state process
\[
  X_t^u=x+\int_0^t u_s\,ds,
  \qquad
  0\le t\le T,
\]
so as to remain close to a prescribed progressively measurable target process
\(\xi=(\xi_t)_{0\le t\le T}\), while penalizing rapid adjustments through the
quadratic cost parameter \(\kappa>0\). In the constrained version, the
controller is additionally required to satisfy
\[
  X_T^u=\Xi_T
  \qquad \mathbb P\text{-a.s.}
\]
for a given random target position \(\Xi_T\) at time \(T\). Our main interest
is the small-friction regime \(\kappa\downarrow0\), in which we investigate how
fast the optimal tracking cost vanishes as adjustment costs disappear. In our
optimal execution application, the running target \(\xi\) will be the
inventory of the optimal execution strategy in a generalized Obizhaeva--Wang model. The constrained tracking
problem will quantify the cost of replacing this strategy, which generally has infinite variation and includes block trades, by a smoother strategy while preserving the
random terminal position~\(\Xi_T\).

For fixed \(\kappa>0\), semi-explicit solutions for both the unconstrained and
constrained versions of~\eqref{def:tracking-intro} were derived by Bank,
Soner, and Vo{\ss}~\cite{bank.al.17} in terms of certain kernels. These
formulas provide useful insights into the stochastic tracking
problem. It is, however, not straightforward to derive general properties, such as asymptotic bounds that can be read from regularity properties of the target (or at least we have not succeeded in doing so).

Our first main contribution is to derive constructive, non-asymptotic upper
bounds for the value function~\eqref{def:tracking-intro}. For the unconstrained problem, the bound is expressed in
terms of an exponentially weighted average of the squared \(L^2\)
time-translation modulus of the target and immediately yields a generic
\(O(\sqrt{\kappa})\) rate whenever the associated seminorm
in~\eqref{eq:seminorm} is finite. We show that this includes
all square-integrable c\`adl\`ag semimartingales, as well as certain non-semimartingale
targets such as fractional Brownian motion with Hurst parameter \(H>1/2\). In general, a lower bound shows that whenever the modulus behaves approximately like a power law, it fully determines the convergence rate up to multiplicative constants. 
We show that the generic square-root rate is sharp for continuous
martingales and for certain deterministic targets with
jumps. By contrast, for fractional Brownian motion with \(H<1/2\), the
time-translation seminorm is infinite and~\eqref{def:tracking-intro} vanishes at the
slower rate \(\kappa^H\).

For the constrained problem with prescribed terminal state $\Xi_T$, the upper bound is
additionally governed by the short-horizon mismatch between the running
target \(\xi\) and the terminal-target martingale
\[
  \Xi_t:=\E[\Xi_T\mid\cF_t],
  \qquad 0\le t\le T,
\]
as well as by the local cost of learning \(\Xi_T\) near maturity. The latter
is related to the reachability criterion identified
in~\cite{bank.al.17} as the necessary and sufficient admissibility condition for a random
terminal position. Again, the typical rate is \(O(\sqrt{\kappa})\) in the most relevant cases, and sharpness is established in benchmark examples.
Importantly, all our upper bounds are constructive. In the unconstrained
problem, the bound follows from a suboptimal state $X^u$ that exponentially filters the target at relaxation
scale \(\sqrt{\kappa}\). In the constrained problem, this tracker is
concatenated with a terminal bridge that enforces the random terminal
position. This construction will also be used to produce an explicit, implementable execution strategy in the second part of the paper.

Indeed, our second main contribution concerns optimal execution with a stochastic
terminal target. The model is motivated by the optimal execution problem of
a central risk book (CRB) as in Nutz, Webster, and
Zhao~\cite{NutzWebsterZhao.26}. A CRB receives random order flow while
unwinding its inventory, so that the total quantity to be executed is learned
over time and the terminal constraint is inherently stochastic. If the order flow is truth-telling (cf.~\cite{NutzWebsterZhao.26}), the time-$t$ inventory is precisely $\Xi_t$. We study a
generalized Obizhaeva--Wang model~\cite{ObizhaevaWang.13} of transient price impact with deterministic time-varying resilience and depth, random terminal position \(\Xi_T\), and
c\`adl\`ag semimartingale execution strategies. The central financial
question is the cost of suppressing the block trades and infinite variation of the optimal singular
strategy by imposing absolute continuity and a quadratic penalty on the
trading rate. This question was first posed by Ulrich Horst in
connection with the regularized model that was used in~\cite{NutzWebsterZhao.26}.

We first solve the unregularized execution problem explicitly by a direct
Hilbert-space argument. This yields the unique optimal c\`adl\`ag
semimartingale strategy without relying on the backward stochastic
differential equation (BSDE) techniques, and the corresponding mathematical restrictions,
employed in previous works (see also below). The
optimal inventory is driven by the terminal-target martingale
\((\Xi_t)_{0\le t\le T}\) and therefore typically has infinite variation, in addition to block trades. To produce an absolutely continuous optimal execution strategy, the problem is often regularized by adding an instantaneous
quadratic penalty on the trading rate (as in, for example, Graewe and
Horst~\cite{GraeweHorst.17}), forcing execution strategies to avoid spikes in
trading activity that would leak too much information to the market.

Our Hilbert-space formulation provides the connection between the
execution and tracking parts of the paper. Writing $J_0(Q)$ for the execution cost of an inventory strategy $Q$ in the (unregularized) price impact model and \(Q^0\) for the optimal inventory, the excess cost of $Q$ is precisely the squared distance between the corresponding impact processes $Y^Q$ and $Y^{Q^0}$ in the Hilbert space~\(\cH\) that we introduce:
\[
  J_0(Q)-J_0(Q^0)
  =
  \|Y^Q-Y^{Q^0}\|_{\cH}^2.
\]
By the impact dynamics of the Obizhaeva--Wang model, the right-hand side is further bounded by a constant times $\E\int_0^T |Q_t-Q_t^0|^2\,dt$. As a result, the approximation of the regularized execution problem reduces to constrained tracking of the
unregularized optimizer \(Q^0\) by an absolutely continuous strategy~$Q$ with terminal position $Q_T=Q_T^0$.

We verify that \(Q^0\) has a finite squared \(L^2\) time-translation seminorm
and satisfies a uniform \(L^2\) bound relative to the terminal-target
martingale. These estimates allow us to specialize the constructive tracking
mechanism developed in the first part of the paper. The resulting strategy
\(Q^\eps\)  exponentially tracks \(Q^0\) with relaxation scale
\(\sqrt{\eps}\) until the deterministic time \(T-\sqrt{\eps}\), and then
switches to a terminal bridge enforcing \(Q_T^\eps=\Xi_T\). The feedback form~\eqref{eq:Qeps-rate} of \(Q^\eps\) is simple, explicit, and implementable; it does not require solving the less tractable regularized optimization problem.

We derive non-asymptotic bounds showing that the excess price
impact cost
\(
  J_0(Q^\eps)-J_0(Q^0)
\)
is of order \(O(\sqrt{\eps})\). Thus, \(Q^\eps\) is an implementable,
nearly optimal alternative to the generally unknown regularized optimizer. Indeed, it approximates $J_0(Q^0)$ at the same rate as the true optimizer: we also show that the difference between the regularized cost of $Q^\eps$ and the unregularized value function is of order \(O(\sqrt{\eps})\), which is a stronger result and implies that the difference between the value functions of the regularized and unregularized problems is of order \(O(\sqrt{\eps})\). Finally, we show that this rate is sharp in the benchmark case of constant model parameters and deterministic terminal position.

\bigskip

Our paper sits at the intersection of two strands of the literature. The
first concerns optimal tracking of a general target
process with absolutely continuous controls and a random terminal state
constraint, as introduced by Bank, Soner, and
Vo\ss{}~\cite{bank.al.17}, whose setup we adopt. Bank and
Vo\ss{}~\cite{BankVoss.18} study a more general framework with stochastic
coefficients. Related examples of targeting problems and terminal-state
constraints appear in Ankirchner and Kruse~\cite{annkirchner.kruse.15},
Graewe and Horst~\cite{GraeweHorst.17}, and Dolinsky, Gottesman, and
Gurel-Gurevich~\cite{DolinskyGottesmanGurelGurevich.20}, among others. These
works, however, do not study the optimal tracking of a general
stochastic target process. 
Compared with~\cite{bank.al.17}, our contribution
differs in emphasis: rather than seeking exact solution formulas, we derive
non-asymptotic bounds tied directly to the roughness of the target and, in the
constrained problem, to the information flow associated with the random
terminal position. Cai, Rosenbaum, and Tankov derive asymptotic lower bounds
for a class of tracking objectives and control types
in~\cite{CaiRosenbaumTankov.17} and construct asymptotically optimal feedback
strategies for many examples in the companion paper
\cite{CaiRosenbaumTankov.17feedback}. Their framework is restricted to
continuous It\^o-semimartingale targets and does not include random terminal
state constraints, and therefore does not apply directly to our problem.

The second strand concerns transient price impact and optimal execution. Variants of the Obizhaeva--Wang model~\cite{ObizhaevaWang.13} with
deterministic time-varying illiquidity parameters were analyzed, among others,
by Bank and Fruth~\cite{BankFruth.14}, Alfonsi and
Acevedo~\cite{AlfonsiAcevedo.14}, and Fruth, Sch\"oneborn, and
Urusov~\cite{FruthSchonebornUrusov.13}. Extensions with stochastic liquidity
include Fruth, Sch\"oneborn, and Urusov~\cite{FruthSchonebornUrusov.19} and
Ackermann, Kruse, and
Urusov~\cite{AckermannKruseUrusov.21disc,AckermannKruseUrusov.21cont,AckermannKruseUrusov.22}.
In contrast to these works, we derive the optimal strategy for our model
through a direct Hilbert-space argument rather than through BSDEs. See also Webster~\cite{Webster.23} and Cartea, Jaimungal, and Penalva~\cite{CarteaJaimungalPenalva.15} for extensive references.

The regularized problem with quadratic instantaneous costs is related to the
optimal liquidation problems studied by Graewe and
Horst~\cite{GraeweHorst.17}, Nutz, Webster, and
Zhao~\cite{NutzWebsterZhao.26}, and Chen, Horst, and
Tran~\cite{ChenHorstTran.25}. The latter obtain an explicit solution in the
benchmark case of constant coefficients and a deterministic terminal
position. Following~\cite{NutzWebsterZhao.26}, we cover deterministic time-varying coefficients and a general
stochastic terminal target. 
The work of Horst and Kivman~\cite{HorstKivman.24} is comparable to our asymptotic analysis; they study the convergence of the optimal execution strategy in the optimal liquidation framework of~\cite{GraeweHorst.17} when the instantaneous cost parameter vanishes. The convergence result is qualitative, it holds in probability and uniformly in time over time intervals bounded away from the initial and terminal times. At $t=0$ and $t=T$, this convergence fails as the limiting strategy has jumps. They show that the limiting c\`adl\`ag semimartingale strategy solves an optimal liquidation problem of Obizhaeva–Wang-type as in~\cite{AckermannKruseUrusov.21cont} with a deterministic liquidation constraint, but the speed of convergence is not considered. Our focus is on random terminal
targets and quantitative results, both for the optimizer and the directly implementable strategy that we propose. In a different direction, related asymptotic questions are also considered
by Dolinskyi and Dolinsky~\cite{DolinskyiDolinsky.24}, who study
exponential-utility indifference prices of European options under a
liquidation constraint in a regime of small linear price impact and high risk
aversion.

\bigskip

The remainder of the paper is organized as follows.
\Cref{sec:unconstrained} studies the asymptotics of the unconstrained tracking
problem, and \cref{sec:constrained} treats the version with a random terminal
state constraint. In \cref{sec:ow}, we introduce and explicitly solve the
unregularized Obizhaeva--Wang execution problem. Finally,
\cref{sec:ow-asymptotics} combines the Hilbert-space structure of the
execution problem with the tracking mechanisms from
\cref{sec:unconstrained,sec:constrained} to derive an explicit absolutely
continuous execution strategy with non-asymptotic cost guarantees, as well as the sharp
square-root convergence rate for the regularized value.

\section{Unconstrained optimal tracking}\label{sec:unconstrained}

Fix $T>0$, a filtered probability space $(\Omega,\mathcal{F},(\mathcal{F}_t)_{0\le t\le T},\mathbb{P})$ satisfying the usual conditions, and an initial position $x\in\R$. For any process $u$ in the class of admissible controls
\begin{equation}\label{eq:controls}
\sU
=
\left\{
u=(u_t)_{0\le t\le T}:\;
u \text{ is progressively measurable and }
\E\int_0^T u_t^2\,dt<\infty
\right\},
\end{equation}
we define the associated state process
\begin{equation} \label{def:stateprocess}
X_t^u=x+\int_0^t u_s\,ds,
\qquad
0\le t\le T.
\end{equation}
Throughout, we fix a progressively measurable target process
\(\xi=(\xi_t)_{0\le t\le T}\) satisfying
\begin{equation}\label{eq:xi-L2}
\xi_0\in L^2(\mathbb{P}),
\qquad
\E\int_0^T \xi_t^2\,dt<\infty,
\end{equation}
and consider the quadratic tracking objective 
\begin{equation}\label{eq:objective}
J_\kappa(u)
=
\E\left[
\frac12\int_0^T (X_t^u-\xi_t)^2\,dt
+
\frac{\kappa}{2}\int_0^T u_t^2\,dt
\right],
\qquad
\kappa\ge 0.
\end{equation}
The associated value is
\begin{equation}\label{eq:value-function}
v(\kappa):=\inf_{u\in\sU}J_\kappa(u).
\end{equation}

Our aim is to quantify the behavior of $v(\kappa)$ as $\kappa\downarrow0$, which clearly depends on the regularity of the target~$\xi$. Our key insight is that the relevant roughness can be quantified by the following (squared) $L^2$ time-translation modulus of $\xi$:
\begin{equation}\label{eq:omega}
\omega_\xi(h)
:=
\E\int_0^T \left|\xi_t-\xi_{(t-h)^+}\right|^2\,dt,
\qquad
h>0,
\end{equation}
where $(t-h)^+:=\max\{t-h,0\}$. By definition,
\begin{align*}
\omega_\xi(h) & = \E\int_0^h \left|\xi_t-\xi_{0}\right|^2\,dt + \E\int_h^T \left|\xi_t-\xi_{t-h}\right|^2 dt, \qquad  0 < h \leq T,  \\ 
\omega_\xi(h) & = \omega_\xi(T) = \E\int_0^T \left|\xi_t-\xi_{0}\right|^2\,dt, \qquad
h\ge T,
\end{align*}
and $\lim_{h \downarrow 0} \omega_\xi(h) = 0$. In particular, note that $\omega_\xi(h) < \infty$ for all $h > 0$, for any target process $\xi$ satisfying~\eqref{eq:xi-L2}. Moreover, we introduce
\begin{equation} \label{eq:seminorm}
   \mysnorm :=\sup_{h>0}\frac{\omega_\xi(h)}{h}    
\end{equation}
and refer to it as the squared $L^2$ time-translation seminorm of $\xi$. %

The following theorem provides a general upper bound on the value~\eqref{eq:value-function} in terms of an exponentially weighted average of the time-translation modulus $\omega_\xi(\cdot)$ of $\xi$.

\begin{theorem}\label{thm:master}
For every $\kappa>0$,
\begin{equation}\label{eq:master}
v(\kappa)
\le
\sqrt{\kappa}\,\E[(x-\xi_0)^2]
+
\frac{2}{\sqrt{\kappa}}
\int_0^\infty e^{-h/\sqrt{\kappa}}\,\omega_\xi(h)\,dh.
\end{equation}
\end{theorem}

\begin{proof}
The idea of the proof is to construct a control $u \in\sU$ such that $J_\kappa(u)$ is bounded by the right-hand side of~\eqref{eq:master}. It turns out that this is achieved by a myopic control which steers $X^u_t$ toward $\xi_t$ at every instant $t \in [0,T]$ at a suitably chosen rate.

Specifically, set $a=\sqrt{\kappa}$, consider the progressively measurable process $X$ defined by
\begin{equation}\label{eq:benchmarkX}
X_t
=
x+\frac1a\int_0^t (\xi_s-X_s)\,ds,
\qquad
0\le t\le T,
\end{equation}
and define
\[
u_t:=\frac1a(\xi_t-X_t), \qquad 0\le t\le T.
\]
The linear equation \eqref{eq:benchmarkX} has the explicit solution
\begin{equation}\label{eq:Xexplicit}
X_t
=
e^{-t/a}x+\frac1a\int_0^t e^{-(t-s)/a}\xi_s\,ds.
\end{equation}
By Jensen's inequality and \eqref{eq:xi-L2}, we have \(X\in L^2(\Omega\times[0,T])\), hence \(u\in\sU\) and \(X=X^u\).

Consider the error
\[
e_t:=X_t-\xi_t.
\]
Since \(u_t=-a^{-1}e_t\) and \(a^2=\kappa\), the tracking objective in~\eqref{eq:objective} can be written as
\begin{equation}\label{eq:cost-equals-error}
J_\kappa(u)=\E\int_0^T e_t^2\,dt.
\end{equation}

Moreover, let \(\mu_a(dh)=a^{-1}e^{-h/a}\,dh\), a probability measure on \([0,\infty)\).
By the change of variables \(h=t-s\) in \eqref{eq:Xexplicit}, we can write
\[
X_t
=
e^{-t/a}x+\int_0^t \xi_{t-h}\,\mu_a(dh).
\]
Since \(\int_t^\infty \mu_a(dh)=e^{-t/a}\), we obtain
\begin{align} 
e_t = & \, e^{-t/a} x + \int_0^t \xi_{t-h} \, \mu_a(dh) - \int_0^{+\infty} \xi_t\,\mu_a(dh) \nonumber \\
= & \, e^{-t/a}(x-\xi_0) - \int_0^\infty \bigl(\xi_t-\xi_{(t-h)^+}\bigr)\,\mu_a(dh). \label{eq:error}    
\end{align}
Hence, by \((r+s)^2\le 2r^2+2s^2\) and Jensen's inequality,
\[
|e_t|^2
\le
2e^{-2t/a}(x-\xi_0)^2
+
2\int_0^\infty \left|\xi_t-\xi_{(t-h)^+}\right|^2\,\mu_a(dh).
\]
Integrating over \(t\in[0,T]\), taking expectations, and applying Fubini's
theorem yield
\begin{align*}
\E\int_0^T e_t^2\,dt
&\le
2\E[(x-\xi_0)^2]\int_0^T e^{-2t/a}\,dt
+
2\int_0^\infty \omega_\xi(h)\,\mu_a(dh)\\
&\le
a\,\E[(x-\xi_0)^2]
+
\frac{2}{a}\int_0^\infty e^{-h/a}\omega_\xi(h)\,dh.
\end{align*}
Combining this estimate with \eqref{eq:cost-equals-error} and
\(v(\kappa)\le J_\kappa(u)\) proves \eqref{eq:master}.
\end{proof}

\begin{remark}\label{rem:initial-match}
If $x=\xi_0$ almost surely, the proof of \cref{thm:master} gives the
sharper estimate
\begin{equation} \label{eq:sharper}
v(\kappa)
\le
\frac{1}{\sqrt{\kappa}}
\int_0^\infty e^{-h/\sqrt{\kappa}}\,\omega_\xi(h)\,dh.    
\end{equation}
Indeed, in that case, the initial displacement term vanishes, and Jensen's inequality can be applied directly in~\eqref{eq:error} to omit the factor two.
\end{remark}

The next corollary shows that the $L^2$ time-translation seminorm of $\xi$ in~\eqref{eq:seminorm} provides the natural constant associated with the generic rate $\sqrt{\kappa}$ in the limiting behavior of $v(\kappa)$. We write $\Gamma(\cdot)$ for the usual Gamma function.

\begin{corollary}\label{cor:rates}
Suppose that, for some constants \(L>0\) and \(\alpha>0\),
\begin{equation}\label{eq:holder-modulus}
\omega_\xi(h)\le Lh^\alpha,
\qquad
h>0.
\end{equation}
Then, for every \(\kappa>0\),
\begin{equation}\label{eq:rate-alpha}
v(\kappa)
\le
\sqrt{\kappa}\,\E[(x-\xi_0)^2]
+
2L\,\kappa^{\alpha/2}\,\Gamma(\alpha+1).
\end{equation}
Consequently, for \(0<\kappa\le1\),
\begin{equation}\label{eq:rate-min-alpha-one}
v(\kappa)
\le
\kappa^{(\alpha\wedge1)/2}\left(\E[(x-\xi_0)^2]+2L\,\Gamma(\alpha+1)\right).
\end{equation}
In particular, if \( \mysnorm <\infty\), it holds that
\begin{equation}\label{eq:rate-sqrtkappa}
v(\kappa)
\le
\sqrt{\kappa}\,\bigl(\E[(x-\xi_0)^2]+2 \mysnorm \bigr)
\end{equation}
for all \(\kappa>0\). Finally, if \(x=\xi_0\) a.s., we have the sharper estimate
\begin{equation}\label{eq:rate-alpha-initial-match}
v(\kappa)
\le
L\kappa^{\alpha/2}\,\Gamma(\alpha+1), \qquad \kappa>0.
\end{equation}
\end{corollary}

\begin{proof}
Insert \eqref{eq:holder-modulus} into \eqref{eq:master}. With the change of
variables \(h=\sqrt{\kappa}\,r\),
\[
\frac1{\sqrt{\kappa}}
\int_0^\infty e^{-h/\sqrt{\kappa}}h^\alpha\,dh
=
\kappa^{\alpha/2}\int_0^\infty e^{-r}r^\alpha\,dr
=
\Gamma(\alpha+1)\kappa^{\alpha/2}.
\]
This proves \eqref{eq:rate-alpha}. For \(0<\kappa\le1\), the two terms in
\eqref{eq:rate-alpha} are both bounded by the right-hand side of
\eqref{eq:rate-min-alpha-one}, depending on whether \(\alpha < 1\) or
\(\alpha \geq 1\). Moreover, \eqref{eq:rate-sqrtkappa} is the special case with
\(L=\mysnorm\) and \(\alpha=1\). Lastly, if \(x=\xi_0\) a.s., the sharper estimate in~\eqref{eq:sharper} from~\cref{rem:initial-match} applies, and the same change of variables from above yields \eqref{eq:rate-alpha-initial-match}. 
\end{proof}

The preceding upper bounds admit a complementary lower bound in terms of the same time-translation modulus, showing that under mild scale regularity of \(\omega_\xi\), the modulus determines the convergence rate up to multiplicative constants.

\begin{proposition}[Lower bound]
\label{prop:general-lower-bound}
For every \(0<\kappa\le T^2/2\),
\begin{equation}\label{eq:general-lower-bound}
  v(\kappa)
  \ge
  \frac{1}{12}\,
  \omega_\xi\!\left(\sqrt{2\kappa}\right)
  -
  \frac{\sqrt{2}}{4}\sqrt{\kappa}\,
  \E\bigl[(x-\xi_0)^2\bigr].
\end{equation}
In particular, if \(x=\xi_0\) almost surely, then
\begin{equation}\label{eq:tracking-sandwich}
  \frac{1}{12}\,
  \omega_\xi\!\left(\sqrt{2\kappa}\right)
  \le
  v(\kappa)
  \le
  \frac{1}{\sqrt{\kappa}}
  \int_0^\infty
  e^{-h/\sqrt{\kappa}}\omega_\xi(h)\,dh,
  \qquad
  0<\kappa\le T^2/2.
\end{equation}
If \(x=\xi_0\) almost surely and 
\begin{equation}\label{eq:two-sided-modulus}
  c h^\alpha
  \le
  \omega_\xi(h)
  \le
  C h^\alpha,
  \qquad
  0<h\le T
\end{equation}
for some \(c,C,\alpha>0\), then
\begin{equation}\label{eq:two-sided-rate}
  \frac{c\,2^{\alpha/2}}{12}\,
  \kappa^{\alpha/2}
  \le
  v(\kappa)
  \le
  C\Gamma(\alpha+1)\kappa^{\alpha/2},
  \qquad
  0<\kappa\le T^2/2
\end{equation}
and thus
\[
  v(\kappa)
  \asymp
  \kappa^{\alpha/2}
  \asymp
  \omega_\xi(\sqrt{\kappa}),
  \qquad
  \kappa\downarrow0.
\]
\end{proposition}

\begin{proof}
Fix \(u\in\sU\), and write
\[
  X_t:=x+\int_0^t u_s\,ds,
  \qquad
  e_t:=X_t-\xi_t, \qquad A:=\E\int_0^T e_t^2\,dt,
  \qquad
  B:=\E\int_0^T u_t^2\,dt,
\]
so that
\(
  J_\kappa(u)=\frac12(A+\kappa B).
\)
Fix \(0<h\le T\). Since
\[
  \xi_t-\xi_{(t-h)^+}
  =
  \bigl(X_t-X_{(t-h)^+}\bigr)
  -
  e_t
  +
  e_{(t-h)^+},
\]
the inequality \((a+b+c)^2\le3(a^2+b^2+c^2)\) gives
\begin{align*}
  \omega_\xi(h)
  &\le
  3\,\E\int_0^T
  |X_t-X_{(t-h)^+}|^2\,dt
  +
  3A
  +
  3\,\E\int_0^T e_{(t-h)^+}^2\,dt.
\end{align*}
Because \(e_0=x-\xi_0\), a change of variables yields
\begin{align*}
  \E\int_0^T e_{(t-h)^+}^2\,dt
  &=
  h\,\E[(x-\xi_0)^2]
  +
  \E\int_0^{T-h} e_t^2\,dt
  \le
  h\,\E[(x-\xi_0)^2]+A.
\end{align*}
Moreover, by Cauchy--Schwarz and Fubini,
\begin{align*}
  \E\int_0^T
  |X_t-X_{(t-h)^+}|^2\,dt
  &=
  \E\int_0^T
  \left|
    \int_{(t-h)^+}^t u_s\,ds
  \right|^2dt
  \le
  h\,\E\int_0^T
  \int_{(t-h)^+}^t u_s^2\,ds\,dt
  \le
  h^2 B.
\end{align*}
Consequently,
\begin{equation}\label{eq:omega-lower-proof}
  \omega_\xi(h)
  \le
  6A+3h^2B
  +
  3h\,\E[(x-\xi_0)^2].
\end{equation}

Now let \(0<\kappa\le T^2/2\) and choose
\(h=\sqrt{2\kappa}\). Then \(h\le T\), and
\[
  6A+3h^2B
  =
  6A+6\kappa B
  =
  12J_\kappa(u).
\]
Thus, \eqref{eq:omega-lower-proof} implies
\[
  J_\kappa(u)
  \ge
  \frac{1}{12}\,
  \omega_\xi\!\left(\sqrt{2\kappa}\right)
  -
  \frac{\sqrt{2}}{4}\sqrt{\kappa}\,
  \E[(x-\xi_0)^2].
\]
Taking the infimum over \(u\in\sU\) proves
\eqref{eq:general-lower-bound}. 
If \(x=\xi_0\) almost surely, the lower bound in
\eqref{eq:tracking-sandwich} follows immediately, while the upper bound is
\eqref{eq:sharper} from \cref{rem:initial-match}.

Finally, assume \eqref{eq:two-sided-modulus}. The lower bound in
\eqref{eq:two-sided-rate} follows from
\eqref{eq:general-lower-bound}. Since
\(\omega_\xi(h)=\omega_\xi(T)\) for \(h\ge T\), the upper bound in
\eqref{eq:two-sided-modulus} extends to all \(h>0\). Hence
\eqref{eq:rate-alpha-initial-match}, with \(L=C\), gives the upper bound in
\eqref{eq:two-sided-rate}. The final equivalences follow from
\eqref{eq:two-sided-modulus}.
\end{proof}

\subsection{Examples and special cases}\label{sec:examples-satisfying}

We first show that every square-integrable c\`adl\`ag semimartingale target process $\xi$ has a finite $L^2$ time-translation seminorm, as defined in~\eqref{eq:seminorm}, and thus has the generic $\sqrt{\kappa}$ rate established in~\eqref{eq:rate-sqrtkappa}.   

\begin{proposition}[Semimartingale target]\label{prop:semimart}
Suppose
\[
\xi_t=\xi_0+M_t+A_t,
\qquad
0\le t\le T,
\]
where \(M\) is a square-integrable c\`adl\`ag martingale with \(M_0=0\) and
\(A\) is an adapted c\`adl\`ag finite-variation process with \(A_0=0\), such
that
\[
\E[\langle M\rangle_T]+\E[|A|_T^2]<\infty.
\]
Then, for every \(0<h\le T\),
\begin{equation}\label{eq:semimart-bound}
\omega_\xi(h)
\le
2h\,\E[\langle M\rangle_T]
+
2h\,\E[|A|_T^2].
\end{equation}
In particular, \( \mysnorm <\infty\), and therefore 
\(v(\kappa)=O(\sqrt{\kappa})\).
\end{proposition}

\begin{proof}
Write
\[
\xi_t-\xi_{(t-h)^+}
=
\bigl(M_t-M_{(t-h)^+}\bigr)+\bigl(A_t-A_{(t-h)^+}\bigr).
\]
Hence
\[
\omega_{\xi}(h)\le 2\omega_M(h)+2\omega_A(h),
\]
where
\[
\omega_M(h):=\E\int_0^T \left|M_t-M_{(t-h)^+}\right|^2\,dt,
\qquad
\omega_A(h):=\E\int_0^T \left|A_t-A_{(t-h)^+}\right|^2\,dt.
\]

For the martingale part, using Fubini's theorem and the properties of the predictable quadratic variation $\langle M \rangle$ of $M$ yields 
\begin{equation} \label{eq:martingale-increment}
\begin{aligned}
\omega_M(h)
&=
\E\int_0^h M_t^2\,dt
+
\E\int_h^T (M_t-M_{t-h})^2\,dt\\
&=
\E\int_0^h \langle M\rangle_t\,dt
+
\E\int_h^T \bigl(\langle M\rangle_t-\langle M\rangle_{t-h}\bigr)\,dt\\
&=
\E\int_{T-h}^T \langle M\rangle_t\,dt
\le
h\,\E[\langle M\rangle_T].
\end{aligned}
\end{equation}
For the finite-variation part, let \(V_t:=|A|_t\) and note that
\[
|A_t-A_{(t-h)^+}|
\le
\int_{((t-h)^+,\,t]} dV_s
=:D_t^{(h)}.
\]
Moreover, since \(D_t^{(h)}\le V_T\), we obtain via Fubini's theorem that
\begin{equation} \label{eq:finite-variation-increment}
\begin{aligned}
\int_0^T |A_t-A_{(t-h)^+}|^2\,dt
&\le
V_T\int_0^T D_t^{(h)}\,dt
=
V_T\int_0^T\int_{((t-h)^+,\,t]} dV_s\,dt\\
&=
V_T\int_0^T\int_s^{(s+h)\wedge T} dt\,dV_s
\le
hV_T^2.
\end{aligned}
\end{equation}
Taking expectations and applying once more Fubini's theorem yields \(\omega_A(h)\le h\,\E[|A|_T^2]\). Combining the two estimates proves \eqref{eq:semimart-bound}. Lastly, since \(\omega_\xi(h)=\omega_\xi(T)\) for
\(h\ge T\), this also implies \(\mysnorm<\infty\), and thus $v(\kappa)=O(\sqrt{\kappa})$ by virtue of~\eqref{eq:rate-sqrtkappa}.
\end{proof}

Note that the bounds of~\cref{thm:master,cor:rates} are one-sided. To see that the $\sqrt{\kappa}$ rate for semimartingale target processes (\cref{prop:semimart}) is sharp in specific cases, it is convenient to use the exact formula of Bank et al.~\cite[Theorem~3.1, formula~(8)]{bank.al.17} for fixed~$\kappa$. This formula simplifies substantially in the two regimes considered below: continuous martingale targets and deterministic finite-variation targets.

\begin{proposition}[Continuous martingale target]\label{prop:martingale-sharp}
Assume that
\[
\xi_t=\xi_0+M_t,
\qquad
0\le t\le T,
\]
where \(\xi_0\in\R\) is deterministic and \(M\) is a continuous
square-integrable martingale with \(M_0=0\). Then
\begin{equation}\label{eq:martingale-exact}
v(\kappa)
=
\frac{\sqrt{\kappa}}{2}\tanh\!\left(\frac{T}{\sqrt{\kappa}}\right)(x-\xi_0)^2
+
\frac12\,\E\int_0^T \sqrt{\kappa}\,
\tanh\!\left(\frac{T-t}{\sqrt{\kappa}}\right)\,d\langle M\rangle_t.
\end{equation}
In particular,
\begin{equation}\label{eq:martingale-asymptotic}
v(\kappa)
=
\frac{\sqrt{\kappa}}{2}\Bigl((x-\xi_0)^2+\E[\langle M\rangle_T]\Bigr)
+
o(\sqrt{\kappa}),
\qquad
\kappa\downarrow0,
\end{equation}
so that the $\sqrt{\kappa}$ rate is sharp whenever $\xi$ is not identically equal to $x$.
\end{proposition}

\begin{proof}
Continuous adapted processes are predictable, so the exact unconstrained value
formula of Bank et al.~\cite[Theorem~3.1, formula~(8)]{bank.al.17}
applies, and since $\xi$ is a martingale, the formula simplifies to the expression in~\eqref{eq:martingale-exact}. Hence
\begin{equation*} 
    \frac{2 v(\kappa)}{\sqrt{\kappa}} = \tanh\!\left(\frac{T}{\sqrt{\kappa}}\right)(x-\xi_0)^2 + \E\int_0^T \tanh\!\left(\frac{T-t}{\sqrt{\kappa}}\right)\,d\langle M\rangle_t.
\end{equation*}
Since \(\tanh(T/\sqrt{\kappa})\to1\) as $\kappa \downarrow 0$, the first term converges to $(x-\xi_0)^2$. For the second term, since \(0\le \tanh((T-t)/\sqrt{\kappa})\le 1\) for all $0 \leq t \leq T$, \(\tanh((T-t)/\sqrt{\kappa})\to1\) for all $t<T$, and $\langle M \rangle$ is continuous, dominated convergence yields  
\begin{equation*}
    \mathbb E\int_0^T \tanh\!\left(\frac{T-t}{\sqrt{\kappa}}\right) \,d\langle M\rangle_t \longrightarrow \mathbb E\int_0^T d\langle M\rangle_t = \mathbb E[\langle M\rangle_T]
\end{equation*}
for $\kappa \downarrow 0$ and we conclude~\eqref{eq:martingale-asymptotic}.
\end{proof}

\begin{remark}[Brownian motion target]\label{rem:brownian-special}
If \(M=W\) is a standard Brownian motion, then
\(\langle M\rangle_t=t\) for all $0 \leq t \leq T$, and \eqref{eq:martingale-exact} reduces to
\begin{align*}
    v(\kappa) = & \, \frac{\sqrt{\kappa}}{2}\tanh\!\left(\frac{T}{\sqrt{\kappa}}\right)(x-\xi_0)^2 + \frac{\kappa}{2}\log\!\left(\cosh\!\left(\frac{T}{\sqrt{\kappa}}\right)\right) \\
    = & \, \frac{\sqrt{\kappa}}{2}\tanh\!\left(\frac{T}{\sqrt{\kappa}}\right)(x-\xi_0)^2 + \frac{\sqrt\kappa \, T}{2}- \frac{\kappa}{2} \log 2+ \frac{\kappa}{2} \log(1+e^{-2T/\sqrt{\kappa}}).
\end{align*}
Hence, using the fact that $\tanh(T/\sqrt{\kappa}) = 1+O(e^{-2T/\sqrt{\kappa}})$ and $\log(1+e^{-2T/\sqrt{\kappa}}) = O(e^{-2T/\sqrt{\kappa}})$ for $\kappa \downarrow 0$, we obtain
\begin{equation*}
v(\kappa) = \frac{\sqrt{\kappa}}{2} \left((x-\xi_0)^2+T\right) - \frac{\kappa}{2}\log 2 + O\!\left(\sqrt{\kappa}\,e^{-2T/\sqrt{\kappa}}\right), \qquad
\kappa\downarrow 0.   
\end{equation*}
In particular,
\[
v(\kappa)
=
\frac{\sqrt{\kappa}}{2}\bigl((x-\xi_0)^2+T\bigr)
+
O(\kappa),
\qquad
\kappa\downarrow0,
\]
showing most transparently that the $\sqrt{\kappa}$ upper bound has the correct order.
\end{remark}

For the next example, it is notationally convenient to introduce the Sobolev space %
\begin{align*}
H^1(0,T)=\bigg\{ g\in L^2(0,T):\;& g \text{ admits a representative } \tilde g \text{ with }
\tilde g(t)=\tilde g(s)+\int_s^t \dot g(r)\,dr \\
&\text{for all } s,t\in[0,T], \text{ for some } \dot g\in L^2(0,T)\bigg\}.
\end{align*}

\begin{proposition}[Deterministic targets with finitely many jumps]\label{prop:piecewise-sharp}
Assume that the target process is deterministic and of the form
\begin{equation}\label{eq:piecewise-target}
\xi_t
=
\xi_0+g(t)+\sum_{j=1}^N \Delta_j\,\1_{\{t\ge t_j\}},
\qquad
0\le t\le T,
\end{equation}
where \(N\ge 0\), $g \in H^1(0,T)$, \(g(0)=0\), and the
jump times satisfy $0<t_1<\cdots<t_N<T$. Then
\begin{equation}\label{eq:piecewise-asymptotic}
v(\kappa)
=
\sqrt{\kappa}\left(
\frac12(x-\xi_0)^2+\frac14\sum_{j=1}^N \Delta_j^2
\right)
+
o(\sqrt{\kappa}),
\qquad
\kappa\downarrow0.
\end{equation}
In particular, whenever \(x\ne \xi_0\) or at least one jump is nonzero, the
generic $\sqrt{\kappa}$ rate cannot be improved.
\end{proposition}

\begin{proof}
Set \(a:=\sqrt{\kappa}\). Since \(\xi\) differs at most at the deterministic
times \(t_1,\dots,t_N\) from a predictable deterministic process, the exact
unconstrained value formula of Bank et al.~\cite[Theorem~3.1, formula~(8)]{bank.al.17} applies after changing \(\xi\) on a Lebesgue-null set. Because the target process is deterministic, this formula simplifies to
\begin{equation}\label{eq:det-value-bank}
v(\kappa)
=
\frac{a}{2}\tanh\!\left(\frac{T}{a}\right)(x-\hat\xi_0^\kappa)^2
+
\frac12\int_0^T (\xi_t -\hat\xi_t^\kappa)^2\,dt,
\end{equation}
where
\[
\hat\xi_t^\kappa:=\int_t^T \xi_u K_a(t,u)\,du,
\qquad
K_a(t,u):=\frac{\cosh((T-u)/a)}{a\sinh((T-t)/a)},
\quad 0\le t\le u<T.
\]
Note that for all $0 \leq t < T$, the kernel \(K_a(t,\cdot)\) integrates to one over \([t,T]\). Let us decompose \(\xi=g+J\) with
\[
J(t):=\xi_0+\sum_{j=1}^N \Delta_j\,\1_{\{t\ge t_j\}},
\qquad
0\le t\le T.
\]
By linearity, \(\hat\xi^\kappa=\hat g^\kappa+\hat J^\kappa\), where
\(\hat g^\kappa_t:=\int_t^T g(u)K_a(t,u)\,du\) and
\(\hat J^\kappa_t:=\int_t^T J(u)K_a(t,u)\,du\).

For the $H^1$ part, using \(g(t)-g(u)=-\int_t^u \dot g(r)\,dr\) for $u > t$ and
\[
\int_r^T K_a(t,u)\,du
=
\frac{\sinh((T-r)/a)}{\sinh((T-t)/a)}
\le e^{-(r-t)/a},
\qquad
t\le r < T,
\]
we obtain with Fubini's theorem that
\[
|g(t)-\hat g_t^\kappa| = \left| \int_t^T (g(t) - g(u)) K_a(t,u) \, du \right| \leq \int_t^T |\dot g(r)|e^{-(r-t)/a}\,dr.
\]
Writing 
\begin{equation*}
    \int_t^T |\dot g(r)| e^{-(r-t)/a} \,dr = \int_{\mathbb R} k_a(t-r) |\dot g(r)| \,dr = (k_a* |\dot g|)(t)
\end{equation*}
with $k_a(s):=e^{s/a} \mathbf 1_{(-\infty,0]}(s)$ (and extending $\dot g$ by zero outside $[0,T]$), Young's convolution inequality gives
\[
\|g-\hat g^\kappa\|_{L^2(0,T)}
\le \| k_a * |\dot g| \|_{L^2(\mathbb{R})} \leq \| k_a \|_{L^1(\mathbb{R})} \| \dot g\|_{L^2(0,T)} \leq 
a\|\dot g\|_{L^2(0,T)},
\]
and thus
\begin{equation}\label{eq:g-part-o-a}
\int_0^T (g(t)-\hat g_t^\kappa)^2\,dt=O(a^2).
\end{equation}
Similarly,
\[
|\hat g_0^\kappa|
=
\left|\int_0^T g(u)K_a(0,u)\,du\right|
\le
\int_0^T |\dot g(r)|e^{-r/a}\,dr
\le
\sqrt{\frac a2}\,\|\dot g\|_{L^2(0,T)},
\]
where we used Cauchy–Schwarz in the last inequality, and hence
\begin{equation}\label{eq:g-zero-small}
\hat g_0^\kappa=O(a^{1/2}).
\end{equation}

For the pure jump part, a direct computation gives
\begin{align*}
    \hat J_t^\kappa = & \, \int_t^T J(u) K_a(t,u)\,du = \xi_0 + \sum_{j: t \geq t_j} \Delta_j + \sum_{j: t < t_j} \Delta_j \int_{t_j}^T K_a(t,u)\,du,
\end{align*}
and therefore
\[
\hat J_t^\kappa-J(t)
=
\sum_{j = 1}^N \Delta_j q_j^a(t) \1_{\{t < t_j\}},
\qquad
q_j^a(t):=\frac{\sinh((T-t_j)/a)}{\sinh((T-t)/a)}.
\]
Consequently,
\begin{equation} \label{eq:J-square-asymptotic}
    \int_0^T (\hat J_t^\kappa - J(t))^2\,dt = \sum_{j=1}^N \Delta_j^2 \int_{0}^{t_j} (q_j^a(t))^2 dt + 2 \sum_{1 \leq j < k \leq N} \Delta_j \Delta_k \int_0^{t_j} q_j^a(t) q_k^a(t) dt.
\end{equation}
Next, since
\[
q_j^a(t)
=
e^{-(t_j-t)/a}
\frac{1-e^{-2(T-t_j)/a}}{1-e^{-2(T-t)/a}},
\qquad 0\le t<t_j < T,
\]
and
\begin{equation*}
    \frac{1}{1-e^{-2(T-t)/a}} = 1 + \frac{e^{-2(T-t)/a}}{1-e^{-2(T-t)/a}} = 1 + O(e^{-2(T-t)/a}), 
\end{equation*}
we get
\begin{equation*}
    \frac{1-e^{-2(T-t_j)/a}}{1-e^{-2(T-t)/a}} = (1-e^{-2(T-t_j)/a}) \bigl(1 + O(e^{-2(T-t)/a})\bigr) = 1+O(e^{-2(T-t_j)/a})
\end{equation*}
for all $0\leq t \leq t_j$, because $2(T-t_j) \leq 2 (T-t)$. This yields
\[
q_j^a(t)=e^{-(t_j-t)/a}\bigl(1+O(e^{-2(T-t_j)/a})\bigr)
\]
uniformly on \([0,t_j]\) and thus
\begin{align}
\int_0^{t_j} (q_j^a(t))^2\,dt = & \, \bigl(1 + O(e^{-2(T-t_j)/a})\bigr) \int_0^{t_j} e^{-2(t_j-t)/a} \,dt \nonumber 
\\
= & \, \frac{a}{2} (1-e^{-2t_j/a}) \bigl(1 + O(e^{-2(T-t_j)/a})\bigr) = \frac{a}{2} + o(a), \qquad a \downarrow 0. \label{eq:q-square-asymptotic}
\end{align}
Similarly, for \(j<k\) with $t_j < t_k$, we obtain
\begin{equation} \label{eq:qq-square-asymptotic}
\int_0^{t_j} q_j^a(t)q_k^a(t)\,dt
= e^{-(t_k-t_j)/a} \bigl(1 + O(e^{-c/a})\bigr) \int_0^{t_j} e^{-2(t_j-t)/a} \,dt = 
o(a),   
\end{equation}
where $c>0$ is some constant. Plugging~\eqref{eq:q-square-asymptotic} and~\eqref{eq:qq-square-asymptotic} back into~\eqref{eq:J-square-asymptotic} yields
\begin{equation}\label{eq:J-mismatch-asymptotic}
\int_0^T (\hat J_t^\kappa-J(t))^2\,dt
=
\frac a2\sum_{j=1}^N \Delta_j^2+o(a).
\end{equation}
Moreover, note that
\[
\hat J_0^\kappa-\xi_0
=
\sum_{j=1}^N \Delta_j\,\frac{\sinh((T-t_j)/a)}{\sinh(T/a)}
=
o(1), \qquad a \downarrow 0.
\]
Hence, together with~\eqref{eq:g-zero-small}, we can write
\begin{equation}\label{eq:xi-zero-small}
\hat\xi_0^\kappa = \hat{g}_0^\kappa + \hat{J}_0^\kappa - \xi_0 + \xi_0 = \xi_0 + o(1), \qquad a \downarrow 0.
\end{equation}

Finally, we can put everything together. Regarding the second term in~\eqref{eq:det-value-bank}, using $\xi-\hat\xi^\kappa=(J-\hat J^\kappa)+(g-\hat g^\kappa)$,
we obtain
\begin{align}
    \int_0^T (\xi_t - \hat{\xi}^\kappa_t)^2 \, dt = & \, \int_0^T (J(t) - \hat{J}^\kappa_t)^2 \, dt + \int_0^T (g(t) - \hat{g}^\kappa_t)^2 \, dt \nonumber \\
    & + 2 \int_0^T (J(t) - \hat{J}^\kappa_t)(g(t) - \hat{g}^\kappa_t)\, dt 
    \frac a2\sum_{j=1}^N \Delta_j^2+o(a) \label{eq:L2error-determ}
\end{align}
due to~\eqref{eq:J-mismatch-asymptotic},~\eqref{eq:g-part-o-a}, and the Cauchy–Schwarz estimate
\begin{equation}
    \left\vert \int_0^T (J(t) - \hat{J}^\kappa_t)(g(t) - \hat{g}^\kappa_t)\, dt \right\vert \leq \|J-\hat J^\kappa\|_{L^2(0,T)} \|g-\hat g^\kappa\|_{L^2(0,T)} = O(a^{3/2}) = o(a), 
\end{equation}
since \(\|J-\hat J^\kappa\|_{L^2(0,T)}=O(a^{1/2})\) and \(\|g-\hat g^\kappa\|_{L^2(0,T)}=O(a)\). In particular, we observe that the jump component determines the leading $O(a)$ error in~\eqref{eq:L2error-determ}, while the smoother $g$-component contributes only $O(a^2)$, and its interaction with the jump component is $O(a^{3/2})$. Regarding the first term in~\eqref{eq:det-value-bank}, we can write
\begin{equation}
    \frac{a}{2}\tanh\!\left(\frac{T}{a}\right)(x-\hat\xi_0^\kappa)^2 = \frac{a}{2}(x-\xi_0)^2 +o(a),
\end{equation}
because $\tanh(T/a) = 1 + o(1)$ and $(x-\hat{\xi}_0^\kappa)^2=(x-\xi_0)^2+o(1)$ by~\eqref{eq:xi-zero-small} as $a \downarrow 0$, as well as the fact that $a \, o(1) = o(a)$. Thus,
\begin{equation*}
    v(\kappa) = \frac{a}{2}\tanh\!\left(\frac{T}{a}\right)(x-\hat\xi_0^\kappa)^2 + \frac12\int_0^T (\xi_t -\hat\xi_t^\kappa)^2\,dt = \frac{a}{2}(x-\xi_0)^2 + \frac a4\sum_{j=1}^N \Delta_j^2+o(a),
\end{equation*}
which is \eqref{eq:piecewise-asymptotic}.
\end{proof}

\begin{remark}
In the setting of \cref{prop:piecewise-sharp}, if there are no jumps and \(x=\xi_0\), then the choice
\(u_t=\dot g(t)\) for a deterministic, absolutely continuous target process $\xi_t = \xi_0 + g(t)$ with $g \in H^1(0,T)$ tracks the target in~\eqref{eq:objective} exactly, and thus
\[
v(\kappa)\leq\frac{\kappa}{2}\int_0^T \dot \xi_t^2\,dt. 
\]
In particular, $v(\kappa) = O(\kappa)$. Thus, once the initial mismatch is removed, the \(\sqrt{\kappa}\) contribution in \cref{prop:piecewise-sharp} is produced entirely by the jumps.
\end{remark}

Our final example is concerned with the case where the target process follows a fractional Brownian motion with Hurst parameter $H \in (0,1)$. It is well known that the fractional Brownian motion is a semimartingale with respect to its natural filtration only in the Brownian case $H=0.5$. Hence, when $H \neq 0.5$, it serves as a natural example of a non-semimartingale target process which satisfies \eqref{eq:xi-L2}. 

We first observe that the fractional Brownian motion has finite $L^2$ time-translation seminorm, as defined in~\eqref{eq:seminorm}, only if $H \geq 0.5$. Then, we show that the limiting behavior of $v(\kappa)$ is governed by the Hurst parameter $H$ through the scaling parameter $\kappa^H$.

\begin{lemma} \label{lem:fbm}
Let \(B^H=(B_t^H)_{0\le t\le T}\) be a fractional Brownian motion with Hurst
parameter \(H\in(0,1)\), considered in its natural filtration. Its $L^2$ time-translation modulus is given by
\begin{equation} \label{eq:fBMModulus}
    \omega_{B^H}(h) = T h^{2H} - \frac{2H}{2H+1} h^{2H+1}, \qquad 0 < h \leq T.
\end{equation}
Moreover, 
\begin{equation} \label{eq:fBMSemiNorm}
[B^H]_{\rm{tr}}^2 = \sup_{h>0} \frac{\omega_{B^H}(h)}{h} =
\begin{cases}
\infty, & 0<H<\frac12, \\
T, & H=\frac12, \\
\frac{T^{2H}}{2H} \left(1-\frac{1}{4H^2}\right)^{2H-1}, & H > \frac12.
\end{cases}
\end{equation}
\end{lemma}

\begin{proof}
Let \(0<h\le T\). Since \(B_0^H=0\) and 
\(
\E|B_t^H-B_s^H|^2=|t-s|^{2H},
\)
we obtain
\begin{align*}
\omega_{B^H}(h)
&=
\E\int_0^h |B_t^H|^2\,dt
+
\E\int_h^T |B_t^H-B_{t-h}^H|^2\,dt\\
&=
\int_0^h t^{2H}\,dt+(T-h)h^{2H}
=
Th^{2H}-\frac{2H}{2H+1}h^{2H+1}.
\end{align*}
For \(h \ge T\), we get
\[
\omega_{B^H}(h)=\omega_{B^H}(T)=\frac{T^{2H+1}}{2H+1}.
\]
The claim~\eqref{eq:fBMSemiNorm} then follows from direct computations. 
\end{proof}

\begin{proposition}[Fractional Brownian motion]\label{pr:fbm-with-sharpness}
Let \(B^H=(B_t^H)_{0\le t\le T}\) be fractional Brownian motion with Hurst
parameter \(H\in(0,1)\), considered in its natural filtration. Let
\[
\xi_t=B_t^H,\qquad 0\le t\le T,
\qquad\text{and}\qquad x=0.
\]
Then
\[
v(\kappa)\asymp \kappa^H,
\qquad
\kappa\downarrow0.
\]
More explicitly, for \(0<\kappa\le T^2/2\),
\begin{equation} \label{eq:fBM-asymp}
\frac{\kappa^H}{12}
\left(
T\,2^H-\frac{2H}{2H+1}2^{H+\frac12}\sqrt{\kappa}
\right)
\le
v(\kappa)
\le
T\Gamma(2H+1)\kappa^H .    
\end{equation}
\end{proposition}

\begin{proof}
By \cref{lem:fbm}, the time-translation modulus of \(\xi=B^H\) satisfies
\[
  \omega_\xi(h)\le Th^{2H},
  \qquad h>0.
\]
Indeed, this follows directly from \eqref{eq:fBMModulus} for \(0<h\le T\),
whereas for \(h\ge T\),
\[
  \omega_\xi(h)
  =
  \frac{T^{2H+1}}{2H+1}
  \le
  Th^{2H}.
\]
Thus, \eqref{eq:holder-modulus} holds with \(L=T\) and \(\alpha=2H\).
Since \(x=\xi_0=0\), the sharper estimate
\eqref{eq:rate-alpha-initial-match} yields
\[
  v(\kappa)
  \le
  T\Gamma(2H+1)\kappa^H,
  \qquad
  \kappa>0.
\]

For the lower bound, let \(0<\kappa\le T^2/2\). Since \(x=\xi_0\),
\cref{prop:general-lower-bound} gives
\[
  v(\kappa)
  \ge
  \frac{1}{12}\,
  \omega_\xi\!\left(\sqrt{2\kappa}\right).
\]
Moreover, \(\sqrt{2\kappa}\le T\), so \eqref{eq:fBMModulus} implies
\begin{align*}
  v(\kappa)
  &\ge
  \frac{1}{12}
  \left(
    T(2\kappa)^H
    -
    \frac{2H}{2H+1}(2\kappa)^{H+\frac12}
  \right)\\
  &=
  \frac{\kappa^H}{12}
  \left(
    T\,2^H
    -
    \frac{2H}{2H+1}
    2^{H+\frac12}\sqrt{\kappa}
  \right).
\end{align*}
This proves \eqref{eq:fBM-asymp}. Since the expression in parentheses in
the lower bound converges to \(T\,2^H>0\) as \(\kappa\downarrow0\), the
two bounds also yield
\(
  v(\kappa)\asymp\kappa^H.
\)
\end{proof}

\section{Optimal tracking with terminal state constraint}\label{sec:constrained}

In this section, we study the asymptotic behavior of the optimal tracking problem from Section~\ref{sec:unconstrained} with an additional $\mathbb{P}$-almost sure state constraint at terminal time $T$. %

Recall the admissible class \eqref{eq:controls}, the state dynamics of $X^u = (X^u_t)_{0 \leq t \leq T}$ in~\eqref{def:stateprocess} and, for a given progressively measurable target process $\xi$ satisfying~\eqref{eq:xi-L2}, the tracking objective $J_\kappa(u)$ in~\eqref{eq:objective}. In addition, we let \(\Xi_T\in L^2(\mathcal{F}_T)\) denote a predetermined terminal position and define its c\`adl\`ag square-integrable martingale $\Xi = (\Xi_t)_{0\leq t \leq T}$ via
\[
\Xi_t:=\E[\Xi_T\mid\mathcal{F}_t],\qquad 0\le t\le T.
\]
We assume that
\begin{equation}\label{eq:constrained-reachability}
\int_0^T \frac{d\E[\Xi_t^2]}{T-t}<\infty,
\end{equation}
introduce the constrained admissible set
\begin{equation}\label{eq:constrained-admissible}
\sU^{c}
:=
\left\{u\in\sU:\;
X_T^u = x + \int_0^T u_s\,ds = \Xi_T\quad \mathbb{P}\text{-a.s.}
\right\}
\end{equation}
and consider the value of the constrained optimization problem
\begin{equation}\label{eq:constrained-value}
v^{c}(\kappa):=\inf_{u\in\sU^{c}} J_\kappa(u).
\end{equation}

The integral in~\eqref{eq:constrained-reachability} is understood
in the Stieltjes sense, with the convention that any atom of
\(d\E[\Xi_t^2]\) at \(T\) contributes \(+\infty\). The condition~\eqref{eq:constrained-reachability} was derived in Bank et al.~\cite{bank.al.17} as the relevant reachability criterion; indeed, \cite[Lemma~5.4]{bank.al.17} proves that~\eqref{eq:constrained-reachability} is equivalent to \(\sU^{c}\neq\varnothing\), so that it characterizes the
terminal positions that can be reached by absolutely continuous state processes
with square-integrable controls. We note that \eqref{eq:constrained-reachability} rules out a
terminal jump of \(\Xi\), implies \(\Xi_T\in\mathcal F_{T-}\), and may be
interpreted as a condition on the rate at which the terminal target is
revealed as \(t\uparrow T\).

As in the previous section, our aim is to quantify the behavior of $v^c(\kappa)$ as $\kappa\downarrow0$. The relevant objects are the following: For \(0<a\le T/2\), define
\begin{equation}\label{eq:constrained-Omega}
\Omega(a)
:=
a\,\E[(x-\xi_0)^2]
+
\frac{2}{a}\int_0^\infty e^{-h/a}\,\omega_\xi(h)\,dh,
\end{equation}
\begin{equation}\label{eq:constrained-Theta}
\Theta(a)
:=
\E\int_{T-2a}^T |\xi_t-\Xi_t|^2\,dt,
\end{equation}
and
\begin{equation}\label{eq:constrained-Lambda}
\Lambda(a)
:=
\int_{T-2a}^T \frac{d\E[\Xi_t^2]}{T-t}.
\end{equation}
The term \(\Omega(a)\) is the familiar right-hand side of \eqref{eq:master} evaluated at scale \(a\), while \(\Theta(a)\) measures the short-horizon mismatch
between the running target \(\xi\) and the terminal martingale \(\Xi\), and
\(\Lambda(a)\) records the local cost of learning \(\Xi_T\) near maturity. Note that, under our assumptions, we have $\Theta(a) < \infty$ and $\Lambda(a) < \infty$.

The next theorem is the counterpart to~\cref{thm:master}, providing a general upper bound on the constrained value~\eqref{eq:constrained-value}. 

\begin{theorem}\label{thm:constrained-master}
Let \(\xi\) be progressively measurable and satisfy \eqref{eq:xi-L2}, and let
\(\Xi_T\in L^2(\mathcal{F}_T)\) satisfy \eqref{eq:constrained-reachability}.
Then, for every \(\kappa\in(0,T^2/4]\), with \(a:=\sqrt{\kappa}\),
\begin{equation}\label{eq:constrained-master}
v^{c}(\kappa)
\le
6\bigl(\Omega(a)+\Theta(a)+a^2\Lambda(a)\bigr).
\end{equation}
\end{theorem}

\begin{proof}
Fix \(\kappa\in(0,T^2/4]\) and set \(a:=\sqrt{\kappa}\). The idea of the proof is to construct a control $u \in \sU^{c}$ such that $J_\kappa(u)$ is bounded by the right-hand side of~\eqref{eq:constrained-master}. It turns out that this is achieved by the concatenation of a myopic control targeting $\xi$ up to some suitable time $s \in (0,T)$, and then an asymptotically singular control forcing $X^{u}_t$ toward the terminal value $\Xi_T$ by time $T$.

\medskip
\noindent\emph{The myopic phase.}
On \([0,T-a]\), let \(Y\) solve
\[
Y_t
=
x+\frac{1}{a}\int_0^t (\xi_r-Y_r)\,dr,
\qquad
0\le t\le T-a,
\]
and put \(e_t:=Y_t-\xi_t\). Exactly as in the proof of \cref{thm:master}, we have
\[
\E\int_0^{T-a} e_t^2\,dt\le \Omega(a).
\]
Moreover, introducing \(u_t^{(1)}:=a^{-1}(\xi_t-Y_t) \in L^2(\Omega \times [0,T-a])\), we obtain
\[
\frac12\int_0^{T-a} (Y_t-\xi_t)^2\,dt
+
\frac{\kappa}{2}\int_0^{T-a} (u_t^{(1)})^2\,dt
=
\int_0^{T-a} e_t^2\,dt,
\]
and for every \(s\in[T-2a,T-a]\), we can conclude that
\begin{equation}\label{eq:constrained-stage1}
\E\left[
\frac12\int_0^s (Y_t-\xi_t)^2\,dt
+
\frac{\kappa}{2}\int_0^s (u_t^{(1)})^2\,dt
\right]
\le
\Omega(a).
\end{equation}
Now, choose \(s\in[T-2a,T-a]\) such that
\[
\E\bigl[e_s^2+|\Xi_s-\xi_s|^2\bigr]
\le
\frac{1}{a}
\int_{T-2a}^{T-a}\E\bigl[e_t^2+|\Xi_t-\xi_t|^2\bigr]\,dt.
\]
Since
\[
\int_{T-2a}^{T-a}\E[e_t^2]\,dt\le \Omega(a),
\qquad
\int_{T-2a}^{T-a}\E[|\Xi_t-\xi_t|^2]\,dt\le \Theta(a),
\]
it follows that
\begin{equation} \label{eq:proof-B}
\E\bigl[e_s^2+|\Xi_s-\xi_s|^2\bigr]
\le
\frac{\Omega(a)+\Theta(a)}{a}.
\end{equation}
We also introduce the constant 
\begin{equation} \label{eq:proof-constrained-B1}
B:=\E[|\Xi_s-Y_s|^2]
\end{equation}
and note that
\begin{equation} \label{eq:proof-constrained-B2}
B = \E[|\Xi_s-\xi_s + \xi_s - Y_s|^2]
\le
2\,\E\bigl[e_s^2+|\Xi_s-\xi_s|^2\bigr]
\le
\frac{2}{a}\bigl(\Omega(a)+\Theta(a)\bigr).
\end{equation}

\medskip
\noindent\emph{The terminal bridge.}
On \([s,T)\), define \(X\) by
\[
dX_t=\frac{\Xi_t-X_t}{T-t}\,dt,
\qquad
X_s=Y_s.
\]
Its explicit solution is given by
\begin{equation} \label{eq:proof-terminal-state1}
X_t
=
\frac{T-t}{T-s}Y_s
+
(T-t)\int_s^t \frac{\Xi_r}{(T-r)^2}\,dr,
\qquad
s\le t<T.
\end{equation}

We first argue that $\lim_{t \uparrow T} X_t = \Xi_T$ $\mathbb{P}$-a.s. To this end, note that due to the martingale convergence theorem, we have
\begin{equation}
    \lim_{t \uparrow T} \Xi_t = \mathbb{E}[\Xi_T \mid \cF_{T-}] = \Xi_T \qquad \mathbb{P}\text{-a.s.},
\end{equation}
because~\eqref{eq:constrained-reachability} implies $\Xi_{T} \in \cF_{T-}$. Moreover, we can write
\begin{equation} \label{eq:proof-terminal-state2}
X_t-\Xi_T =
\frac{T-t}{T-s}(Y_s-\Xi_T) + (T-t) \int_s^t \frac{\Xi_r-\Xi_T}{(T-r)^2}\,dr, \qquad
s\le t<T,   
\end{equation}
since
\begin{equation*}
(T-t) \int_s^t\frac{1}{(T-r)^2} \,dr = 1 - \frac{T-t}{T-s}.    
\end{equation*}
The first term in~\eqref{eq:proof-terminal-state2} converges to zero as $t \uparrow T$. For the second term, fix an outcome $\omega \in \Omega$ for which $\Xi_t(\omega)\to\Xi_T(\omega)$ for $t \uparrow T$. Given $\varepsilon > 0$, choose $s < u <T$ such that
\begin{equation*}
|\Xi_r(\omega)-\Xi_T(\omega)|\le \varepsilon, \qquad u\le r<T.    
\end{equation*}
Then, for $t > u$, we obtain the estimate
\begin{equation*}
    (T-t) \int_s^t \frac{\vert \Xi_r-\Xi_T \vert}{(T-r)^2}\,dr \leq (T-t) \int_s^u \frac{\vert \Xi_r-\Xi_T \vert}{(T-r)^2}\,dr + \varepsilon (T-t) \int_u^t \frac{1}{(T-r)^2}\,dr.
\end{equation*}
The first term tends to zero as $t \uparrow T$, while the second term is bounded by $\varepsilon$. Hence, $\limsup_{t\uparrow T}|X_t-\Xi_T|\le \varepsilon$ in~\eqref{eq:proof-terminal-state2}. Since $\varepsilon>0$ is arbitrary, we get that $\lim_{t \uparrow T} X_t = \Xi_T$ $\mathbb P$-a.s.

Next, we define
\begin{equation*}
    u_t^{(2)} := \frac{\Xi_t-X_t}{T-t}, \qquad s \le t < T.
\end{equation*}
Using integration by parts in~\eqref{eq:proof-terminal-state1} yields
\begin{equation*}
    X_t = \frac{T-t}{T-s}Y_s + \Xi_t - \frac{T-t}{T-s} \Xi_s - (T-t) \int_{(s,t]} \frac{1}{T-r} \, d\Xi_r, \qquad s \le t<T,
\end{equation*}
and hence
\begin{equation} \label{eq:proof-terminal-state3}
    \Xi_t-X_t = \frac{T-t}{T-s}\bigl(\Xi_s-Y_s\bigr) + (T-t)\int_{(s,t]} \frac{1}{T-r}\,d\Xi_r, \qquad s\le t<T,
\end{equation}
so that we obtain the representation
\begin{equation} \label{eq:proof-u2}
    u_t^{(2)} = \frac{\Xi_s-Y_s}{T-s} + \int_{(s,t]} \frac{1}{T-r}\,d\Xi_r, \qquad s \le t<T.
\end{equation}
By the martingale isometry, Fubini's theorem, and~\eqref{eq:constrained-reachability}, we have
\begin{equation} \label{eq:proof-iso}
\E\int_s^T
\left(
\int_{(s,t]} \frac{1}{T-r}\,d\Xi_r
\right)^2dt
=
\int_{(s,T]} \frac{d\E[\Xi_r^2]}{T-r}
\le
\Lambda(a) < \infty.
\end{equation}
Hence, we can conclude that \(u^{(2)}\in L^2(\Omega\times[s,T])\). In particular, note that the concatenated control
\[
u_t
:=
u_t^{(1)}\mathbf{1}_{[0,s)}(t)+u_t^{(2)}\mathbf{1}_{[s,T)}(t)
\]
belongs to \(\sU^{c}\).

\medskip
\noindent\emph{The total cost.}
We start with estimating the cost on \([s,T]\), where we recall \(s\in[T-2a,T-a]\). First,
\[
|X_t-\xi_t|^2\le 2|X_t-\Xi_t|^2+2|\Xi_t-\xi_t|^2.
\]
Moreover, for \(s\le t<T\), using~\eqref{eq:proof-terminal-state3} and the constant $B$ defined in~\eqref{eq:proof-constrained-B1}, we obtain
\begin{align*}
\E[|\Xi_t-X_t|^2]
&\le
2\frac{(T-t)^2}{(T-s)^2}B
+
2(T-t)^2
\E\left[
\left(
\int_{(s,t]} \frac{1}{T-r}\,d\Xi_r
\right)^2
\right]\\
& \leq
2\frac{(T-t)^2}{(T-s)^2}B
+
2 (T-t)
\int_{(s,t]} \frac{1}{T-r} \,d\mathbb{E}[\Xi_r^2]
\\
& \le
2\frac{(T-t)^2}{(T-s)^2}B
+
2(T-t)\Lambda(a),
\end{align*}
where we used the fact that $r \leq t$, which implies $\frac{T-t}{T-r}\leq 1$ and hence $(\frac{T-t}{T-r})^2 \leq \frac{T-t}{T-r}$. Consequently,
\begin{align*}
\E\int_s^T |X_t-\xi_t|^2\,dt
&\le
2\,\E\int_s^T |X_t-\Xi_t|^2\,dt
+
2\,\Theta(a)\\
&\le
2\int_s^T
\left(
2\frac{(T-t)^2}{(T-s)^2}B
+
2(T-t)\Lambda(a)
\right)\,dt
+
2\,\Theta(a)\\
&=
\frac{4}{3}(T-s)B+2(T-s)^2\Lambda(a)+2\Theta(a)\\
&\le
\frac{8a}{3}B+8a^2\Lambda(a)+2\Theta(a),
\end{align*}
because \(T-s\le 2a\). Hence,
\begin{equation}\label{eq:constrained-tail-tracking}
\frac12\,\E\int_s^T |X_t-\xi_t|^2\,dt
\le
\frac{4a}{3}B+4a^2\Lambda(a)+\Theta(a).
\end{equation}
Next, using~\eqref{eq:proof-u2} and~\eqref{eq:proof-iso} together with the definition of~$B$, we get the upper bound 
\begin{align*}
\E\int_s^T (u_t^{(2)})^2\,dt
&\le
\frac{2}{T-s}B
+
2\,\E\int_s^T
\left(
\int_{(s,t]} \frac{1}{T-r}\,d\Xi_r
\right)^2dt\\
&\le
\frac{2}{a}B+2\Lambda(a),
\end{align*}
since \(T-s\ge a\). Thus, we have
\begin{equation}\label{eq:constrained-tail-trading}
\frac{\kappa}{2}\,\E\int_s^T (u_t^{(2)})^2\,dt
\le
aB+a^2\Lambda(a).
\end{equation}
Combining \eqref{eq:constrained-tail-tracking} and
\eqref{eq:constrained-tail-trading}, the cost on the terminal
interval $[s,T]$ is bounded by
\[
\frac{7a}{3}B+5a^2\Lambda(a)+\Theta(a).
\]
Using \eqref{eq:proof-constrained-B2}, this yields
\[
\E\left[
\frac12\int_s^T (X_t-\xi_t)^2\,dt
+
\frac{\kappa}{2}\int_s^T (u_t^{(2)})^2\,dt
\right]
\le
\frac{14}{3}\Omega(a)
+
\frac{17}{3}\Theta(a)
+
5a^2\Lambda(a).
\]
Finally, adding the contribution from \eqref{eq:constrained-stage1} on $[0,s]$, we obtain
\[
J_\kappa(u)
\le
\frac{17}{3}\Omega(a)
+
\frac{17}{3}\Theta(a)
+
5a^2\Lambda(a)
\le
6\bigl(\Omega(a)+\Theta(a)+a^2\Lambda(a)\bigr).
\]
Since the constructed control satisfies \(u\in\sU^{c}\), this proves \eqref{eq:constrained-master}.
\end{proof}

\begin{remark}\label{rem:constrained-interpretation}

The bound \eqref{eq:constrained-master} separates the constrained cost into
three pieces: the same contribution \(\Omega(a)\) as in the
unconstrained problem, the terminal compatibility defect \(\Theta(a)\), and the
learning term \(a^2\Lambda(a)\). Under the general reachability assumption~\eqref{eq:constrained-reachability}, we have
\(\Lambda(a)\to0\) as \(a\downarrow0\), and hence, in general,
\[
a^2\Lambda(a)=o(a)=o(\sqrt{\kappa})
\qquad\text{when } a=\sqrt{\kappa}.
\]
Thus, at the generic \(\sqrt{\kappa}\) scale, the new effect of the terminal constraint $X^u_T = \Xi_T$ $\mathbb{P}$-a.s.\ is governed by the short-horizon mismatch between the running target $\xi$ and the terminal-target martingale $\Xi$ measured by \(\Theta(a)\).
\end{remark}

The next corollary is analogous to~\cref{cor:rates} in the unconstrained case, and again establishes $\sqrt{\kappa}$ as the generic rate.

\begin{corollary}\label{cor:constrained-rates}
Under the assumptions of \cref{thm:constrained-master}, assume that there
exist constants \(L_1,L_2\ge0\) and exponents \(\alpha,\beta>0\) such that
\begin{equation}\label{eq:constrained-omega-rate}
\omega_\xi(h)\le L_1 h^\alpha,
\qquad
h>0,
\end{equation}
and
\begin{equation}\label{eq:constrained-Theta-rate}
\Theta(a)\le L_2 a^\beta,
\qquad
0<a\le T/2.
\end{equation}
Then, for every \(\kappa\in(0,T^2/4]\),
\begin{equation}\label{eq:constrained-rate-bound}
v^{c}(\kappa)
\le
6\sqrt{\kappa}\,\E[(x-\xi_0)^2]
+
12L_1\,\Gamma(\alpha+1)\,\kappa^{\alpha/2}
+
6L_2\,\kappa^{\beta/2}
+
6\kappa\,\Lambda(\sqrt{\kappa}).
\end{equation}
Consequently,
\[
v^{c}(\kappa)
=
O\!\left(\kappa^{(1\wedge\alpha\wedge\beta)/2}\right),
\qquad
\kappa\downarrow0,
\]
and in particular \(v^{c}(\kappa)=O(\sqrt{\kappa})\) if \(\alpha,\beta\ge 1\).

If, in addition, \(x=\xi_0\) almost surely, then the sharper estimate
\begin{equation}\label{eq:constrained-rate-bound-initial-match}
v^{c}(\kappa)
\le
6L_1\,\Gamma(\alpha+1)\,\kappa^{\alpha/2}
+
6L_2\,\kappa^{\beta/2}
+
6\kappa\,\Lambda(\sqrt{\kappa})
\end{equation}
holds. In particular,
\[
v^{c}(\kappa)
=
O\!\left(\kappa^{(\alpha\wedge\beta\wedge2)/2}\right),
\qquad
\kappa\downarrow0,
\]
and \(v^{c}(\kappa)=o(\sqrt{\kappa})\) if \(\alpha,\beta>1\).
\end{corollary}

\begin{proof}
From \eqref{eq:constrained-Omega} and \eqref{eq:constrained-omega-rate},
with the change of variables \(h=ar\),
\[
\frac{1}{a}\int_0^\infty e^{-h/a}\omega_\xi(h)\,dh
\le
L_1a^\alpha\int_0^\infty e^{-r}r^\alpha\,dr
=
L_1\Gamma(\alpha+1)a^\alpha.
\]
Therefore
\[
\Omega(a)
\le
a\,\E[(x-\xi_0)^2]
+
2L_1\Gamma(\alpha+1)a^\alpha.
\]
Inserting this estimate and \eqref{eq:constrained-Theta-rate} into
\eqref{eq:constrained-master}, with \(a=\sqrt{\kappa}\), gives
\eqref{eq:constrained-rate-bound}.

By \eqref{eq:constrained-reachability}, we have
\(\Lambda(a)\to0\) as \(a\downarrow0\). Hence
\[
\kappa\Lambda(\sqrt{\kappa})=o(\kappa)
=
O\!\left(\kappa^{(1\wedge\alpha\wedge\beta)/2}\right), \quad \kappa \downarrow 0,
\]
and the first asymptotic estimate follows.

If \(x=\xi_0\) a.s., the initial displacement term in the myopic phase
vanishes. More precisely, in the proof of \cref{thm:constrained-master}, using
the sharper estimate from \cref{rem:initial-match} gives the same bound as
\eqref{eq:constrained-master} with \(\Omega(a)\) replaced by
\[
\Omega_0(a):=
\frac1a\int_0^\infty e^{-h/a}\omega_\xi(h)\,dh.
\]
The estimate above gives
\[
\Omega_0(a)\le L_1\Gamma(\alpha+1)a^\alpha.
\]
Substituting this into \eqref{eq:constrained-master}, again with
\(a=\sqrt{\kappa}\), proves
\eqref{eq:constrained-rate-bound-initial-match}, 
and the remaining claims follow.
\end{proof}

\begin{remark}[Lower bound]
 Since the unconstrained value function $v$ is trivially a lower bound for the constrained value function $v^c$, the lower bounds in~\cref{prop:general-lower-bound} immediately extend to the constrained case.
\end{remark}

\subsection{Examples and special cases}\label{sec:constrained-examples}

We first record a direct consequence of~\cref{thm:constrained-master} for target processes $\xi$ with finite $L^2$ time-translation seminorm, which includes square-integrable semimartingales by~\cref{prop:semimart}. Then, we isolate two complementary benchmarks: a mismatch between a running target $\xi$ and a deterministic terminal position $\Xi_T = z \in \mathbb{R}$, which is controlled by $\Theta(a)$, and the exact continuous-martingale case furnished by \cite[Theorem~3.2, formula~(10)]{bank.al.17} with sharp rate $\sqrt{\kappa}$.

\begin{corollary}\label{cor:constrained-matching}
In the setting of \cref{thm:constrained-master}, assume that
\(\mysnorm<\infty\). Then, for every \(\kappa\in(0,T^2/4]\),
\begin{equation}\label{eq:constrained-matching-bound}
v^{c}(\kappa)
\le
6\sqrt{\kappa}\,\bigl(\E[(x-\xi_0)^2]+2\mysnorm\bigr)
+
6\Theta(\sqrt{\kappa})
+
6\kappa\,\Lambda(\sqrt{\kappa}).
\end{equation}
In particular, if
\begin{equation}\label{eq:constrained-Theta-Oa}
\Theta(a)=O(a),
\qquad
a\downarrow0,
\end{equation}
then
\[
v^{c}(\kappa)=O(\sqrt{\kappa}),
\qquad
\kappa\downarrow0.
\]
\end{corollary}

\begin{proof}
By the definition of the seminorm $\mysnorm$ in~\eqref{eq:seminorm}, we have
\(
\omega_\xi(h)\le \mysnorm \, h
\)
for all 
\(
h>0.
\)
Hence \eqref{eq:constrained-Omega} yields
\[
\Omega(a)
\le
a\,\bigl(\E[(x-\xi_0)^2]+2 \mysnorm \bigr),
\qquad
0<a\le T/2.
\]
Insert this estimate into \cref{thm:constrained-master} with
\(a=\sqrt{\kappa}\) to obtain \eqref{eq:constrained-matching-bound}. If
\eqref{eq:constrained-Theta-Oa} holds, then
\(
\Theta(\sqrt{\kappa})=O(\sqrt{\kappa}),
\)
and \cref{rem:constrained-interpretation} yields
\(
\kappa\,\Lambda(\sqrt{\kappa})=o(\sqrt{\kappa}).
\)
Therefore \(v^{c}(\kappa)=O(\sqrt{\kappa})\).
\end{proof}

\begin{remark}\label{rem:constrained-Theta-sufficient}
A simple but quite general sufficient condition for \eqref{eq:constrained-Theta-Oa} is the local uniform
\(L^2\)-boundedness
\begin{equation}\label{eq:constrained-local-L2-boundedness}
\sup_{t\in[T-\delta,T]}\E|\xi_t-\Xi_t|^2<\infty
\end{equation}
for some \(\delta>0\). Indeed, for \(0<a\le \delta/2\),
\[
\Theta(a)
=
\E\int_{T-2a}^T |\xi_t-\Xi_t|^2\,dt
\le
2a\sup_{t\in[T-\delta,T]}\E|\xi_t-\Xi_t|^2.
\]
In particular, observe that exact agreement of \(\xi\) and \(\Xi\) on a terminal time interval is not
needed for the generic \(\sqrt{\kappa}\) rate.
\end{remark}

\begin{example}[Deterministic terminal target]\label{ex:constrained-deterministic}
Let \(\Xi_T=z\in\R\) be deterministic. Then \(\Xi_t\equiv z\) for all $t \in [0,T]$, so
\[
\Lambda(a)=0,
\qquad
\Theta(a)=\E\int_{T-2a}^T |\xi_t-z|^2\,dt.
\]
Hence \cref{thm:constrained-master} reduces to
\[
v^{c}(\kappa)\le 6\bigl(\Omega(a)+\Theta(a)\bigr),
\qquad
a=\sqrt{\kappa}.
\]
Assume, moreover, that \(\xi\) is deterministic, piecewise $H^1(0,T)$, and continuous at \(T\) with \(\xi_T=z\). Then \(\mysnorm<\infty\) by
\cref{prop:semimart} and 
\begin{equation*}
    \Omega(a) \le a\,\bigl((x-\xi_0)^2 + 2 \mysnorm \bigr).   
\end{equation*}
If, in addition, \(\xi\in H^1(T-\delta,T)\) for some \(\delta>0\), then for \(a\le \delta/2\), applying Jensen yields
\[
\Theta(a)
=
\int_{T-2a}^T |\xi_t-z|^2\,dt
\leq \int_{T-2a}^T (T-t) \left( \int_t^T \vert \dot \xi_s \vert^2 ds \right) dt
\le
2a^2\int_{T-2a}^T |\dot\xi_t|^2\,dt
=
o(a^2),
\]
since \(\dot\xi\in L^2(T-\delta,T)\).
Consequently,
\[
v^{c}(\kappa)=O(\sqrt{\kappa}),
\qquad
\kappa\downarrow0.
\]
If one further has \(\xi\in H^1(0,T)\) and \(x=\xi_0\), then the control \(u_t=\dot \xi_t\) belongs to \(\sU^{c}\) and yields the sharper upper bound
\[
v^{c}(\kappa)\le \frac{\kappa}{2}\int_0^T \dot\xi_t^2\,dt.
\]
\end{example}

\begin{example}[Continuous martingale running and terminal targets]\label{rem:constrained-martingale-exact}
Assume that
\[
\xi_t=\Xi_t=\Xi_0+M_t,
\qquad
0\le t\le T,
\]
where \(\Xi_0\in\R\) is deterministic, \(M\) is a continuous
square-integrable martingale with \(M_0=0\), and
\eqref{eq:constrained-reachability} holds. In this case, the exact constrained value formula
of Bank et al.~\cite[Theorem~3.2, formula~(10)]{bank.al.17} yields
\begin{equation}\label{eq:constrained-martingale-exact}
v^{c}(\kappa)
=
\frac{\sqrt{\kappa}}{2}\coth\!\left(\frac{T}{\sqrt{\kappa}}\right)(x-\Xi_0)^2
+
\frac12\,\E\int_0^T \sqrt{\kappa}\,
\coth\!\left(\frac{T-t}{\sqrt{\kappa}}\right)\,d\langle M\rangle_t.
\end{equation}
Then
\begin{equation}\label{eq:constrained-martingale-asymptotic}
v^{c}(\kappa)
=
\frac{\sqrt{\kappa}}{2}\Bigl((x-\Xi_0)^2+\E[\langle M\rangle_T]\Bigr)
+
o(\sqrt{\kappa}),
\qquad
\kappa\downarrow0,
\end{equation}
so that the generic \(\sqrt{\kappa}\) rate is sharp whenever 
$\xi$ is not identically equal to $x$.

To justify \eqref{eq:constrained-martingale-asymptotic}, first note that $\coth(T/\sqrt{\kappa}) = 1 + o(1)$ as $\kappa \downarrow 0$ and hence
\begin{equation} \label{eq:coth}
    \sqrt{\kappa}\,\coth\!\left(\frac{T}{\sqrt{\kappa}}\right) =
\sqrt{\kappa}+o(\sqrt{\kappa}), \qquad \kappa \downarrow 0.
\end{equation}
Next, for \(0\le t<T\),
\[
0\le
\sqrt{\kappa}\,\coth\!\left(\frac{T-t}{\sqrt{\kappa}}\right)-\sqrt{\kappa}
=
\frac{2\sqrt{\kappa}}{e^{2(T-t)/\sqrt{\kappa}}-1}
\le
\frac{\kappa}{T-t},
\]
because \(e^{2x}-1\ge 2x\) for \(x\ge0\). Since
\(
d\E[\Xi_t^2]=d\E[\langle M\rangle_t],
\)
the reachability condition \eqref{eq:constrained-reachability} gives
\[
\int_0^T \frac{d\E[\langle M\rangle_t]}{T-t}<\infty.
\]
Therefore
\[
0\le
\E\int_0^T \Biggl(
\sqrt{\kappa}\,\coth\!\left(\frac{T-t}{\sqrt{\kappa}}\right)-\sqrt{\kappa}
\Biggr)\,d\langle M\rangle_t
\le
\kappa\int_0^T \frac{d\E[\langle M\rangle_t]}{T-t}
=
o(\sqrt{\kappa}),
\]
and hence
\begin{align*}
    \E\int_0^T \sqrt{\kappa}\,\coth\!\left(\frac{T-t}{\sqrt{\kappa}}\right)\,d\langle M\rangle_t = & \, \E \int_0^T \Biggl( \sqrt{\kappa}\,\coth\!\left(\frac{T-t}{\sqrt{\kappa}}\right)-\sqrt{\kappa} \Biggr)\,d\langle M\rangle_t + \sqrt{\kappa} \, \E [\langle M \rangle_T] \\
    = & \, \sqrt{\kappa} \, \E [\langle M \rangle_T] + o(\sqrt{\kappa}),
\end{align*}
which proves \eqref{eq:constrained-martingale-asymptotic} together with~\eqref{eq:coth}.
\end{example}

\section{An Obizhaeva--Wang model and its solution}\label{sec:ow}

In the remaining two sections, we apply our general tracking results to the optimal execution problem that motivated it. Specifically, we consider a generalized Obizhaeva–Wang model of transient price impact with a random target position at the terminal time and semimartingale execution strategies. In this section, we first solve the resulting singular stochastic control problem using a Hilbert-space approach. Then, in~\cref{sec:ow-asymptotics}, we introduce a regularized version of the problem with instantaneous quadratic costs imposed on the trading rate, and study its small-cost asymptotics. We will show that the asymptotic analysis reduces to a constrained tracking problem of the optimal unregularized strategy with terminal state constraint. Consequently, the theory from~\cref{sec:constrained} will allow us to quantify the asymptotic error rate when approximating the optimal singular solution in our original Obizhaeva–Wang model with an absolutely continuous regularized strategy.

We fix the time horizon $T\in(0,\infty)$ and let $(\Omega,\cF,(\cF_{t})_{t\in[0,T]},\mathbb{P})$ be a filtered probability space satisfying the usual assumptions. By convention, $\cF_{0-}$ is the augmentation of the trivial $\sigma$-field. Let $\beta,\lambda: [0,T] \to (0,\infty)$ be deterministic functions with the following properties.

\begin{assumption} \label{ass:beta-lambda}
\begin{enumerate}
    \item $\beta$ is c\`adl\`ag and of finite variation.
    \item $\lambda$ is absolutely continuous and bounded away from zero, and its a.e.\ derivative admits a c\`adl\`ag version $\dot\lambda$ of finite variation. We set $\lambda_{0-} := \lambda_0$ and define 
    \begin{equation*}
        \gamma_t:=\log(\lambda_t), \qquad 0 \leq t \leq T.  
    \end{equation*}
    Then, $\gamma$ is also absolutely continuous and its a.e.~derivative admits the c\`adl\`ag finite-variation version
    \begin{equation} \label{eq:gamma-dot}
   \dot\gamma_t := \frac{\dot\lambda_t}{\lambda_t}, \qquad 0 \leq t \leq T.
    \end{equation}
    \item $2\beta + \dot\gamma$ is bounded away from zero, i.e., there exists $c>0$ such that $2\beta_t+\dot\gamma_t >c$ for every $t \in [0,T]$.
\end{enumerate}
\end{assumption}

We recall that (iii) ensures absence of price manipulation in the model. For later use, we also define the measurable functions
\begin{equation} \label{eq:eta-theta}
    \eta_t := \dfrac{\beta_t+\dot\gamma_t}{2\beta_t+\dot\gamma_t}, \qquad \theta_t := \int_0^t \dfrac{(\beta_s+\dot\gamma_s)^2}{\lambda_s(2\beta_s+\dot\gamma_s)} \,ds, \qquad 0 \leq t \leq T.
\end{equation}

Note that, under~\cref{ass:beta-lambda}, the functions $\beta, \lambda, \dot\lambda, \dot\gamma$ are all bounded on $[0,T]$. Moreover, the function $\eta$ is c\`adl\`ag and of finite variation on $[0,T]$ (in particular, it is bounded), and $\theta,\dot\theta$ are bounded on $[0,T]$.

Let $\cS^2$ denote the set of c\`adl\`ag semimartingales $S$ with $\E [\sup_{0\le t\le T}|S_t|^2 ]<\infty$ and let $\Xi_T\in L^2(\cF_T)$ be the given target position at the terminal time $T$. The set of admissible execution strategies is given by
\[
  \cQ:=\{Q\in\cS^2:\, Q_{0-}=0,\ Q_T=\Xi_T\}
\]
with $Q_{t}$ representing the cumulative number of shares bought or sold by time $t$. For any strategy $Q\in\cQ$, we define the impact process $Y=Y^Q$ as the unique solution of
\begin{equation}\label{eq:OWreversion}
  dY_t=-\beta_t Y_t\,dt+\lambda_t\,dQ_t, \qquad Y_{0-}=y,
\end{equation}
where $y\in\R$ is the given initial impact and is fixed throughout. The impact process $Y$ is interpreted as the deviation process, driven by $Q$, from an unaffected price process $S \in \cS^2$, in the sense that marginal orders are effectively executed at the price $S_-+Y_-$. The two illiquidity parameters in the model, the resilience $\beta$ and the depth (or push factor) $\lambda$, are allowed to vary deterministically over time, similar to~\cite{FruthSchonebornUrusov.13}. 
 
Given the unaffected asset price process $S$, we are interested in minimizing the execution cost,
\begin{equation}\label{eq:exec-cost}
  \inf_{Q\in\cQ} \E\left[\int_0^T Y_{t-}\,dQ_t+\frac12[Y,Q]_T-\int_0^T Q_{t-}\,dS_t\right],
\end{equation}
subject to the terminal constraint $Q_T=\Xi_T$ specifying the total quantity to be executed. Throughout, note that the integral $\int_a^b := \int_{[a,b]}$ is defined over the closed interval, and the quadratic covariation includes jumps at zero. Assuming that $S$ is a sufficiently regular martingale such that the last integral in~\eqref{eq:exec-cost} is a martingale, the optimization reduces to
\begin{equation}\label{eq:ow-objective}
  \inf_{Q\in\cQ} \E\left[\int_0^T Y_{t-}\,dQ_t+\frac12[Y,Q]_T\right].
\end{equation}
This is the form of the problem typically used in the literature, and we take this formulation as our primitive. 
Observe that, despite the absence of $S$, the problem is inherently stochastic due to the random terminal condition $\Xi_T$. This differs from standard optimal execution
problems, where the order size is known upfront. The random target position can be thought of as arising from a stochastic order flow, which is anticipated and executed over time, akin to the setup studied in Nutz, Webster, and Zhao~\cite{NutzWebsterZhao.26}. Its presence also necessitates the consideration of general c\`adl\`ag semimartingale execution strategies $Q \in \cQ$. A detailed discussion and derivation of the cost functionals in~\eqref{eq:exec-cost} and~\eqref{eq:ow-objective} can be found in~\cite{AckermannKruseUrusov.21cont}. 

The next theorem summarizes the optimal execution strategy for~\eqref{eq:ow-objective}.

\begin{theorem}\label{th:OWsolution}
Suppose~\cref{ass:beta-lambda} is satisfied. Let $\Xi_t:=\E[\Xi_T\mid\cF_t]$ for
$t\in\{0-\}\cup[0,T]$, and define the square-integrable martingale
$M=(M_t)_{t\in\{0-\}\cup[0,T]}$ by
\begin{align}
  M_{0-} &:= \frac{\lambda_T}{1+\lambda_T\theta_T}
  \left(\Xi_{0-}+\frac{y}{\lambda_{0}}\right), \label{eq:defMstart}\\
  M_t &:= M_{0-}+\int_0^t
  \frac{\lambda_T}{1+\lambda_T(\theta_T-\theta_s)}\,d\Xi_s,
  \qquad t\in[0,T]. \label{eq:defM}
\end{align}
Equivalently, for $t \in [0,T]$,
\begin{equation}\label{eq:defMalt}
  M_t=
  \frac{\lambda_T}{1+\lambda_T\theta_T}
  \left(\Xi_0+\frac{y}{\lambda_{0}}\right)
  +\int_{(0,t]}
  \frac{\lambda_T}{1+\lambda_T(\theta_T-\theta_s)}\,d\Xi_s.
\end{equation}
Then
\begin{equation} \label{def:Ystar}
  Y_t^*:=\eta_tM_t \quad\text{for } t\in[0,T),
  \qquad
  Y_T^*:=M_T  
\end{equation}
is the price impact process of the unique optimal strategy for \eqref{eq:ow-objective}. The optimal strategy is the c\`adl\`ag semimartingale $Q^* \in \mathcal{Q}$ satisfying $Q^*_{0-} = 0$, $Q^*_{T}=\Xi_{T}$, and
\begin{equation} \label{eq:optQfinal}
\begin{aligned}
Q^*_t = & \, \left( \Xi_{0-} + \frac{y}{\lambda_{0}} \right) \left[ \frac{\lambda_T}{1+\lambda_T \theta_T} \left( \frac{\eta_t}{\lambda_t} + \theta_t \right) - 1 \right] + \Xi_{0-} \\
& \, + \int_0^t \frac{\lambda_T}{1+\lambda_T(\theta_T-\theta_s)} \left( \frac{\eta_t}{\lambda_t} + \theta_t - \theta_s \right) d\Xi_s, \qquad 0 \leq t < T.
\end{aligned}
\end{equation}
Its initial and terminal block trades are given by
\begin{equation} \label{eq:optQjumps_gen} 
\begin{aligned}
\Delta Q_0^* = & \, \frac{\eta_0}{\lambda_0} \frac{\lambda_T}{1+\lambda_T \theta_T} \left(  \Xi_0 + \frac{y}{\lambda_0} \right) - \frac{y}{\lambda_0}, \\ 
\Delta Q_T^* = & \, \frac{1}{\lambda_T} \left( \frac{\beta_{T-}}{2\beta_{T-} + \dot{\gamma}_{T-}} M_T + \eta_{T-} \Delta M_T \right).
\end{aligned} 
\end{equation}
Finally, the optimal execution cost~\eqref{eq:ow-objective} is given by
\begin{equation} \label{eq:optimal-cost}
    \frac12 \, \left(\E\left[\int_0^T \dot\theta_t M_t^2\,dt\right]
    +\frac{\E[M_T^2]}{\lambda_T} -\frac{y^2}{\lambda_0} \right).
\end{equation}
\end{theorem}

\begin{remark}[Special cases]
    Our~\cref{th:OWsolution} generalizes some existing results in the literature. In the case where the terminal position $\Xi_T = x > 0$ is deterministic, equations~\eqref{eq:optQfinal}–\eqref{eq:optimal-cost} retrieve the optimal strategy and cost derived in Fruth et al.~\cite[Theorem 8.4]{FruthSchonebornUrusov.13}. In the case where $\beta > 0$ and $\lambda > 0$ are constants, our formulas~\eqref{eq:optQfinal}–\eqref{eq:optQjumps_gen} reduce to the ones derived in Ackermann et al.~\cite[Example 5.3]{AckermannKruseUrusov.22} for the special case of a Brownian filtration setting. In particular, this includes the original formulas of Obizhaeva and Wang~\cite[Proposition 2]{ObizhaevaWang.13} with deterministic target position and $y=0$.
\end{remark}

\begin{remark}
The optimal solution in~\cref{th:OWsolution} has some intuitive properties. First, note that the optimal strategy $Q^*$ is characterized by making the controlled deviation process $Y^*$ in~\eqref{def:Ystar} a martingale on $(0,T)$ up to scaling by the deterministic function $\eta$, which is in turn determined by the time-varying illiquidity parameters $\beta$ and $\lambda$. This generalizes an observation from~\cite{ObizhaevaWang.13} for the deterministic Obizhaeva--Wang model with constant $\beta, \lambda$. There, the optimal strategy keeps the deviation process flat over $(0,T)$ by offsetting incurred price impact $\lambda$ against resilience $\beta$. Here, loosely speaking, any execution strategy $Q \in \cQ$ that induces a drift in the rescaled deviation process $Y^Q/\eta$ in~\eqref{eq:OWreversion} cannot be optimal. It can similarly be interpreted as optimally counterbalancing, now in expectation, the illiquidity effects while trading toward the random target~$\Xi_T$. Second, as also observed in~\cite[Example 5.3]{AckermannKruseUrusov.22}, the optimal execution strategy on $(0,T)$ in~\eqref{eq:optQfinal} decomposes into two additive parts: (i) the optimal deterministic (not necessarily absolutely continuous) execution strategy for known order sizes, where the unknown terminal position $\Xi_T$ is simply replaced by its expected value $\Xi_{0-} =\E[\Xi_T\mid\cF_{0-}]$; (ii) random fluctuations around the deterministic strategy, driven by the c\`adl\`ag terminal-target martingale $\Xi$, which incorporates updates about the terminal position $\Xi_T$. In particular, note that jumps in the martingale $\Xi$ during the execution interval $(0,T)$ trigger block trades in the execution strategy~$Q^*$.
\end{remark}

\subsection{A Hilbert-space proof of \cref{th:OWsolution}}

The rest of the section is devoted to proving \cref{th:OWsolution}. As observed in \cite[Lemma~8.6]{FruthSchonebornUrusov.13}, the equation for the transient price impact \eqref{eq:OWreversion} can be inverted pathwise. More precisely, fixing $Q_{0-}=0$ and $Y_{0-}=y$, one has
\begin{equation}\label{eq:OWreversionInverse}
  dQ_t=\frac{\beta_t}{\lambda_t}Y_t\,dt+\frac{1}{\lambda_t}\,dY_t, \qquad Q_{0-}=0,
\end{equation}
and hence, using integration by parts, the representation
\begin{equation}\label{eq:Qexplicit}
  Q_t=\frac{Y_t}{\lambda_t}-\frac{y}{\lambda_0}+\int_0^t \frac{\beta_s+\dot\gamma_s}{\lambda_s}Y_s\,ds, \qquad 0 \leq t \leq T.
\end{equation}

It will be beneficial to cast our control problem~\eqref{eq:ow-objective} purely in terms of the impact process~$Y$.

\begin{lemma}\label{le:reformulation}
Let $Q\in\cQ$ and let $Y$ be the associated impact process. Then
\begin{align}
  \E\left[\int_0^T Y_{t-}\,dQ_t+\frac12[Y,Q]_T\right]
  &= \E\left[\frac12\left(\frac{Y_T^2}{\lambda_T}-\frac{y^2}{\lambda_0}\right)
      +\int_0^T \frac{2\beta_t+\dot\gamma_t}{2\lambda_t}Y_t^2\,dt\right] \nonumber\\
  &= \langle Y,Y\rangle_{\cH}-\frac{y^2}{2\lambda_0}, \label{eq:costReformulation}
\end{align}
where
\begin{equation}\label{eq:defInnerProd}
  \langle Y,Z\rangle_{\cH}
  := \E\left[\int_0^T \frac{2\beta_t+\dot\gamma_t}{2\lambda_t}Y_tZ_t\,dt
     + \frac{1}{2\lambda_T}Y_TZ_T\right].
\end{equation}
Moreover, the terminal constraint $Q_T=\Xi_T$ is equivalent to
\begin{equation}\label{eq:terminalConstraintNew}
  \frac{Y_T}{\lambda_T}-\frac{y}{\lambda_0}+\int_0^T \frac{\beta_t+\dot\gamma_t}{\lambda_t}Y_t\,dt=\Xi_T.
\end{equation}
\end{lemma}

\begin{proof}
Using \eqref{eq:OWreversionInverse}, the identity $Y_{t-}=Y_t$ $dt$-a.e., It\^o's formula, and integration by parts yields
\begin{align*}
  \int_0^T Y_{t-}\,dQ_t+\frac12[Y,Q]_T
  &= \int_0^T Y_{t-}\left(\frac{1}{\lambda_t}\,dY_t+\frac{\beta_t}{\lambda_t}Y_t\,dt\right)
     +\frac12\int_0^T \frac{1}{\lambda_t}\,d[Y,Y]_t\\
  &= \int_0^T \frac{1}{\lambda_t}\left(Y_{t-}\,dY_t+\frac12\,d[Y,Y]_t\right)
     +\int_0^T \frac{\beta_t}{\lambda_t}Y_t^2\,dt\\
  &= \frac12\int_0^T \frac{1}{\lambda_t}\,d(Y_t^2)
     +\int_0^T \frac{\beta_t}{\lambda_t}Y_t^2\,dt\\
  &= \frac12\left(\frac{Y_T^2}{\lambda_T}-\frac{y^2}{\lambda_0}
     -\int_0^T Y_t^2\,d\!\left(\frac{1}{\lambda_t}\right)\right)
     +\int_0^T \frac{\beta_t}{\lambda_t}Y_t^2\,dt\\
  &= \frac12\left(\frac{Y_T^2}{\lambda_T}-\frac{y^2}{\lambda_0}\right)
     +\int_0^T \frac{\dot\lambda_t}{2\lambda_t^2}Y_t^2\,dt
     +\int_0^T \frac{\beta_t}{\lambda_t}Y_t^2\,dt,
\end{align*}
which is \eqref{eq:costReformulation}. The terminal constraint is \eqref{eq:Qexplicit} at time $T$.
\end{proof}

Let $\cR$ be the progressive $\sigma$-field on $\Omega \times [0,T]$ and define the finite measures
\[
  \mu(dt) := \frac{2\beta_t+\dot\gamma_t}{2\lambda_t}\,dt + \frac{1}{2\lambda_T}\delta_T(dt),
  \qquad
  dR(\omega,t):=\mathbb{P}(d\omega)\,\mu(dt).
\]
Then
\[
  \cH:=L^2(\Omega\times[0,T],\cR,R)
\]
is the Hilbert space of all progressively measurable processes $Y=(Y_t)_{t \in [0,T]}$ with finite norm
\begin{equation} \label{eq:H-norm}
    \|Y\|_{\cH} := \sqrt{\langle Y,Y\rangle_{\cH}}   
\end{equation}
induced by the inner product defined in \eqref{eq:defInnerProd}. We also introduce the map $G:\cH\to L^2(\cF_T)$ defined as
\[
  G(Y):=\frac{Y_T}{\lambda_T}+\int_0^T \frac{\beta_t+\dot\gamma_t}{\lambda_t}Y_t\,dt,
  \qquad Y\in\cH.
\]

The next result shows that our control problem in~\eqref{eq:ow-objective} can be seen as a simple projection in the Hilbert space $\cH$ onto a suitable subspace, which is determined by the terminal constraint through the map $G$.

\begin{proposition}\label{pr:projection}
The map $G:\cH\to L^2(\cF_T)$ is linear, continuous, and
\[
  \cY:=\{Y\in\cH:\, G(Y)=\Xi_T+\tfrac{y}{\lambda_0}\}
\]
is a nonempty closed affine subspace of $\cH$. Consequently, the problem
\begin{equation} \label{eq:minimizer-projection}
  \inf_{Y\in\cY}\langle Y,Y\rangle_{\cH}
\end{equation}
admits a unique minimizer $Y^*\in\cY$. Writing
\[
  \cJ:=\ker G=\left\{J\in\cH:\,
  \frac{J_T}{\lambda_T}=-\int_0^T \frac{\beta_t+\dot\gamma_t}{\lambda_t}J_t\,dt\right\},
\]
this minimizer is characterized by the first-order condition
\begin{equation}\label{eq:FOC}
  \langle Y^*,J\rangle_{\cH}=0, \qquad J\in\cJ.
\end{equation}
\end{proposition}

\begin{proof}
Clearly, the map $G$ is linear. For $Y\in\cH$, Cauchy–Schwarz gives
\begin{align*}
  \E\left|\int_0^T \frac{\beta_t+\dot\gamma_t}{\lambda_t}Y_t\,dt\right|^2 = & \, \E\left|\int_0^T \frac{\beta_t+\dot\gamma_t}{\sqrt{\lambda_t} \sqrt{2\beta_t+\dot\gamma_t}} \frac{ \sqrt{2\beta_t+\dot\gamma_t}}{\sqrt{\lambda_t}} Y_t\,dt\right|^2 \\
  \le & \, 2\theta_T\,\E\left[\int_0^T \frac{2\beta_t+\dot\gamma_t}{2\lambda_t}Y_t^2\,dt\right] 
\le 2\theta_T \|Y\|_{\cH}^2.
\end{align*}
Moreover, 
\begin{equation*}
  \E\left[\left|\frac{Y_T}{\lambda_T}\right|^2\right] \le \frac{2}{\lambda_T}\,\|Y\|_{\cH}^2.
\end{equation*}
Hence, $\|G(Y)\|_{L^2(\cF_T)} \leq c \|Y\|_{\cH}$ for some constant $c > 0$, and $G$ is continuous. Therefore, $\cY=G^{-1}(\Xi_T+y/\lambda_0)$ is a closed affine subspace of $\cH$. It is nonempty because the element
\[
  Y_t^\circ:=\mathbf{1}_{\{T\}}(t)\lambda_T\left(\Xi_T+\frac{y}{\lambda_0}\right), \qquad 0 \leq t \leq T,
\]
belongs to $\cY$. Consequently, the projection theorem in Hilbert spaces yields the existence of a unique minimizer $Y^*\in\cY$, and the usual characterization of projections onto affine subspaces gives the condition in~\eqref{eq:FOC}.
\end{proof}

The next step is to derive an explicit formula for $Y^*$ from~\cref{pr:projection}.

\begin{proposition}\label{pr:minimizerH}
Suppose~\cref{ass:beta-lambda} is satisfied. Let $M$ be given by \eqref{eq:defMstart}--\eqref{eq:defM}. Then $M$ is a square-integrable c\`adl\`ag martingale, and the unique minimizer from \cref{pr:projection} is given by the square-integrable c\`adl\`ag semimartingale
\begin{equation} \label{eq:YstarL2}
    Y_t^*=\eta_t M_t, \qquad t\in[0,T), \qquad Y_T^*=M_T.
\end{equation}
\end{proposition}

\begin{proof}
First, observe that~\eqref{eq:defM} defines a square-integrable martingale on $[0,T]$. Indeed, 
the deterministic integrand
\[
  \frac{\lambda_T}{1+\lambda_T(\theta_T-\theta_t)}, \qquad 0 \leq t \leq T,
\]
is bounded since $0\le \theta_T-\theta_t\le \theta_T$ by the definition of $\theta$ in~\eqref{eq:eta-theta}. Moreover, since $\eta$ is deterministic, bounded, c\`adl\`ag, and of finite variation by~\cref{ass:beta-lambda}, the process $Y^*$ defined in~\eqref{eq:YstarL2} is c\`adl\`ag and a square-integrable semimartingale.

Next, we argue that $Y^*$ is an element of $\cY$. First, since
\[
  \frac{2\beta_t+\dot\gamma_t}{2\lambda_t}(Y^*_t)^2 = \frac{2\beta_t+\dot\gamma_t}{2\lambda_t}(\eta_t M_t)^2
  = \frac12\,\dot\theta_t M_t^2
  \quad\text{a.e.},
\]
we obtain that
\[
  \|Y^*\|_{\cH}^2
  = \frac12\,\E\left[\int_0^T \dot\theta_t M_t^2\,dt\right]
    +\frac{1}{2\lambda_T}\E[M_T^2]
  <\infty,
\]
because $\dot\theta$ is bounded by~\cref{ass:beta-lambda}, and hence $Y^* \in \cH$. 

Regarding the condition $G(Y^*) = \Xi_T + y/\lambda_0$, we introduce the function
\[
  a_t:=\frac{1+\lambda_T(\theta_T-\theta_t)}{\lambda_T}
  = \frac{1}{\lambda_T}+\theta_T-\theta_t, \qquad 0 \leq t \leq T.
\]
Clearly, $a$ is continuous and of finite variation with $da_t=-\dot\theta_t\,dt$, and \eqref{eq:defM} is equivalent to
\[
  a_t\,dM_t = d\Xi_t.
\]
Thus, integration by parts yields
\begin{equation} \label{eq:aM}
  d(a_tM_t)=a_t\,dM_t+M_t\,da_t=d\Xi_t-\dot\theta_t M_t\,dt.
\end{equation}
Due to the definition of $M_{0-}$ in~\eqref{eq:defMstart}, we have $a_0M_{0-}=\Xi_{0-}+y/\lambda_0$, and hence, by integrating~\eqref{eq:aM}, we obtain for every $t\in[0,T]$ the representation
\begin{equation}\label{eq:barxiCond}
  \Xi_t
  = \left(\frac{1}{\lambda_T}+\theta_T-\theta_t\right)M_t
    -\frac{y}{\lambda_0}
    +\int_0^t \dot\theta_s M_s\,ds.
\end{equation}
In particular, at $t=T$, it holds that
\begin{equation}\label{eq:terminalM}
  \Xi_T
  = \frac{M_T}{\lambda_T}-\frac{y}{\lambda_0}+\int_0^T \dot\theta_s M_s\,ds,
\end{equation}
which implies 
\begin{equation} \label{eq:terminal}
\begin{aligned} 
    G(Y^*) = & \, \frac{Y^*_T}{\lambda_T}+\int_0^T \frac{\beta_t+\dot\gamma_t}{\lambda_t}Y^*_t\,dt = \frac{M_T}{\lambda_T}+\int_0^T \frac{\beta_t+\dot\gamma_t}{\lambda_t} \eta_t M_t\,dt \\
    = & \, \frac{M_T}{\lambda_T}+\int_0^T \dot\theta_s M_s\,ds = \Xi_T + \frac{y}{\lambda_0}, 
\end{aligned}
\end{equation}
so $Y^*\in\cY$.

It remains to verify the first-order condition in~\eqref{eq:FOC}. To this end, let $J\in\cJ$. By the representation of $Y^*$ in~\eqref{eq:YstarL2}, we get
\begin{equation} \label{eq:FOC1}
  \langle Y^*,J\rangle_{\cH}
  = \frac12\,\E\left[\int_0^T \frac{\beta_t+\dot\gamma_t}{\lambda_t}M_tJ_t\,dt
     +\frac{M_TJ_T}{\lambda_T}\right].  
\end{equation}
Since $J\in\ker G$, it satisfies
\[
  \frac{J_T}{\lambda_T}
  = -\int_0^T \frac{\beta_t+\dot\gamma_t}{\lambda_t}J_t\,dt.
\]
Plugging this back into~\eqref{eq:FOC1} yields
\begin{align*}
  \langle Y^*,J\rangle_{\cH}
  &= \frac12\,\E\left[\int_0^T \frac{\beta_t+\dot\gamma_t}{\lambda_t}M_tJ_t\,dt
     - M_T\int_0^T \frac{\beta_t+\dot\gamma_t}{\lambda_t}J_t\,dt\right]\\
  &= \frac12\,\E\left[\int_0^T \frac{\beta_t+\dot\gamma_t}{\lambda_t}
     \bigl(M_t-\E[M_T\mid\cF_t]\bigr)J_t\,dt\right]
   = 0,
\end{align*}
because $M$ is a martingale. Thus \eqref{eq:FOC} holds, so \cref{pr:projection} identifies $Y^*$ in~\eqref{eq:YstarL2} as the unique minimizer of~\eqref{eq:minimizer-projection}.
\end{proof}

We are now ready to prove~\cref{th:OWsolution}.

\begin{proof}[Proof of \cref{th:OWsolution}]
By~\cref{pr:minimizerH}, the process $Y^* \in \cY$ in~\eqref{eq:YstarL2} represents the unique minimizer of~\eqref{eq:minimizer-projection}. 
Moreover, formula \eqref{eq:Qexplicit} defines the associated c\`adl\`ag semimartingale inventory process $Q^*$ given by
\begin{equation}\label{eq:optQ}
  Q_t^*
  = \frac{Y^*_t}{\lambda_t} - \frac{y}{\lambda_0} + \int_0^t \frac{\beta_s + \dot\gamma_s}{\lambda_s} Y^*_s \, ds =
  \begin{cases}
  \displaystyle
  \frac{\eta_tM_t}{\lambda_t}
  -\frac{y}{\lambda_0}
  +\int_0^t \dot\theta_sM_s\,ds,
  & 0\le t<T,\\[2.2ex]
  \displaystyle
  \frac{M_T}{\lambda_T}
  -\frac{y}{\lambda_0}
  +\int_0^T \dot\theta_sM_s\,ds,
  & t=T,
  \end{cases}
\end{equation}
with $Q_{0-}^*=0$. Observe that $Q^* \in \cS^2$ by~\cref{pr:minimizerH}. Furthermore, since $Y^*\in\cY$, the condition in~\eqref{eq:terminal} yields $Q_T^*=\Xi_T$. Thus $Q^*\in\cQ$. 

Let $Q\in\cQ$ be arbitrary and let $Y$ be its impact process. By \cref{le:reformulation},
\[
  \E\left[\int_0^T Y_{t-}\,dQ_t+\frac12[Y,Q]_T\right]
  = \langle Y,Y\rangle_{\cH}-\frac{y^2}{2\lambda_0}
  \ge \langle Y^*,Y^*\rangle_{\cH}-\frac{y^2}{2\lambda_0},
\]
because $Y^*$ minimizes $\langle Y,Y\rangle_{\cH}$ over $\cY$. Hence $Q^*$ is optimal.

To prove uniqueness, let $\widetilde Q\in\cQ$ be another optimizer and $\widetilde Y$ its impact process. Then $\widetilde Y$ is another minimizer in $\cY$, so $\widetilde Y=Y^*$ in $\cH$. In particular, $\widetilde Y_T=Y_T^*$ a.s.\ and $\widetilde Y_t=Y_t^*$ for $\mathbb{P}\otimes dt$-a.e.\ $(\omega,t)\in\Omega\times[0,T)$. Since both processes are c\`adl\`ag, they are indistinguishable. Then \eqref{eq:Qexplicit} implies that $\widetilde Q$ and $Q^*$ are also indistinguishable.

To justify the representation of the optimal inventory $Q^*$ in~\eqref{eq:optQfinal}, we first insert~\eqref{eq:defM} into~\eqref{eq:optQ} to obtain
\begin{align*}
    Q_t^* = & \, \frac{\eta_t}{\lambda_t} \frac{\lambda_T}{1+\lambda_T\theta_T}
  \left(\Xi_{0-}+\frac{y}{\lambda_{0}}\right) + \frac{\eta_t}{\lambda_t} \int_0^t
  \frac{\lambda_T}{1+\lambda_T(\theta_T-\theta_s)}\,d\Xi_s \\
   & \, -\frac{y}{\lambda_0} + \int_0^t \dot\theta_s \frac{\lambda_T}{1+\lambda_T\theta_T}
  \left(\Xi_{0-}+\frac{y}{\lambda_{0}}\right) ds + \int_0^t \dot\theta_s \left( \int_0^s
  \frac{\lambda_T}{1+\lambda_T(\theta_T-\theta_r)}\,d\Xi_r \right) \, ds \\
  = & \, \left(\Xi_{0-}+\frac{y}{\lambda_{0}}\right) \left(\frac{\eta_t}{\lambda_t} \frac{\lambda_T}{1+\lambda_T\theta_T} + \frac{\lambda_T}{1+\lambda_T\theta_T} \theta_t -1 \right) + \Xi_{0-} \\
  & + \frac{\eta_t}{\lambda_t} \int_0^t
  \frac{\lambda_T}{1+\lambda_T(\theta_T-\theta_s)}\,d\Xi_s + \int_0^t \left( \int_0^s
  \frac{\lambda_T}{1+\lambda_T(\theta_T-\theta_r)}\,d\Xi_r \right) \, d\theta_s, \qquad 0 \leq t < T.
\end{align*}
Performing integration by parts in the last integral (note that $\theta$ is absolutely continuous) yields
\begin{equation*}
    \int_0^t \left( \int_0^s
  \frac{\lambda_T}{1+\lambda_T(\theta_T-\theta_r)}\,d\Xi_r \right) \, d\theta_s = \theta_t \int_0^t
  \frac{\lambda_T}{1+\lambda_T(\theta_T-\theta_s)}\,d\Xi_s - \int_0^t \theta_s \frac{\lambda_T}{1+\lambda_T(\theta_T-\theta_s)}\,d\Xi_s.
\end{equation*}
Inserting this back into the previous equation results in~\eqref{eq:optQfinal}. 

For the initial and final block trades in~\eqref{eq:optQjumps_gen}, we note that
\begin{align*}
    \Delta Q_0^* = \frac{1}{\lambda_0} \Delta Y_0^* = \frac{1}{\lambda_0} (Y_0^* - Y^*_{0-}) = \frac{1}{\lambda_0} (\eta_0 M_0 - y)
\end{align*}
and
\begin{align*}
    \Delta Q_T^* = & \, \frac{1}{\lambda_T} \Delta Y_T^* = \frac{1}{\lambda_T} (Y_T^* - Y^*_{T-}) 
    =\frac{1}{\lambda_T} (M_T - \eta_{T-} M_{T-}) %
    =\frac{1}{\lambda_T} (M_T (1 - \eta_{T-}) + \eta_{T-} \Delta M_T).
\end{align*}
Finally, for the optimal cost in~\eqref{eq:optimal-cost}, we use~\cref{le:reformulation} and obtain
\begin{align*}
    \E\left[\int_0^T Y^*_{t-}\,dQ^*_t+\frac12[Y^*,Q^*]_T\right]
    = & \, \langle Y^*,Y^*\rangle_{\cH}-\frac{y^2}{2\lambda_0} \\
    = & \, \E\left[\int_0^T \frac{2\beta_t+\dot\gamma_t}{2\lambda_t}(Y^*_t)^2\,dt
    + \frac{1}{2\lambda_T}(Y^*_T)^2 \right] -\frac{y^2}{2\lambda_0} \\
    = & \, \frac12 \E\left[\int_0^T \dot \theta_t M_t^2\,dt \right]
    + \frac{1}{2\lambda_T} \E\left[ M_T^2 \right] - \frac{y^2}{2\lambda_0},
\end{align*}
which completes the proof.
\end{proof}

\section{Regularized execution and its approximation rate}\label{sec:ow-asymptotics}

As seen in~\cref{th:OWsolution}, the optimal strategy for the unregularized Obizhaeva–Wang model naturally has jumps. While block trades at the initial and terminal times arise even in the deterministic case, the terminal-target martingale $(\Xi_t)_{0 \leq t \leq T}$ in our setting typically also leads to non-zero quadratic variation during the execution interval $(0,T)$. A common approach to smooth the strategy is to use a regularized version of the model where an instantaneous quadratic cost penalizes fast trading and in particular imposes absolute continuity of the execution strategy. If we consider the instantaneous cost as an auxiliary regularizer to produce better-behaved strategies, this regularity comes at a price, namely, additional execution cost. %

While the standard approach would be to use the regularized optimal strategy as execution strategy, that optimizer is not explicitly available in our setting with stochastic terminal condition and time-varying coefficients. Below, we construct a \emph{simple, explicit, nearly optimal} strategy $Q^\eps$ which achieves a similar cost as the optimizer. Specifically, our main result in \cref{th:OW-rate} gives a non-asymptotic bound on the additional execution cost, which turns out to be of order $\sqrt{\eps}$ for both strategies. We further show in \cref{pr:OW-rate-sharp} that this order is sharp even in the benchmark deterministic case.

Let $J_0$ denote the objective of the unregularized problem in~\eqref{eq:ow-objective} from Section~\ref{sec:ow}, which can be rewritten as
\[
  J_0(Q):=\langle Y^Q,Y^Q\rangle_{\cH}-\frac{y^2}{2\lambda_0} 
\]
according to~\cref{le:reformulation}. We also introduce the corresponding value
\begin{equation} \label{eq:opt-unregularized}
  V(0):=\inf_{Q\in\cQ}J_0(Q)=J_0(Q^0)
\end{equation}
where \(Q^0:=Q^* \in \cQ\) denotes the unique optimizer from \cref{th:OWsolution} with price impact process \(Y^0:=Y^{Q^0}\). 

Next, we introduce the regularized problem. Let $\varepsilon > 0$ and let \(\cQ_{\mathrm{ac}} \subset \cQ\) denote the subclass of absolutely continuous strategies \(Q\) of the form
\[
  Q_t=\int_0^t u_s\,ds,\qquad 0\le t\le T,
\]
with progressively measurable \(u\) satisfying \(\E\int_0^T u_t^2\,dt<\infty\),
and \(Q_T=\Xi_T \in L^2(\cF_T) \). For any \(Q\in\cQ_{\mathrm{ac}}\) with impact process $Y^Q$ given by~\eqref{eq:OWreversion}, we introduce the objective $J_\eps$ by adding an instantaneous quadratic cost,
\begin{equation} \label{eq:objective-regularized}
  J_\eps(Q):=J_0(Q)+\eps\,\E\int_0^T \dot Q_t^2\,dt.
\end{equation}
We further denote the corresponding value by
\begin{equation} \label{eq:opt-regularized}
  V(\eps) := \inf_{Q\in\cQ_{\mathrm{ac}}}J_\eps(Q).
\end{equation}

Throughout this section, \cref{ass:beta-lambda} is in force. Moreover, for the square-integrable terminal-target martingale $\Xi_t =\E[\Xi_T\mid\cF_t]$, $t\in[0,T]$, we assume that the reachability condition in~\eqref{eq:constrained-reachability} from~\cref{sec:constrained} is satisfied and denote its value by
\begin{equation}\label{eq:XiReachability}
  \Lambda_{\Xi}:=\int_0^T \frac{d\E[\Xi_t^2]}{T-t}<\infty.
\end{equation}

\begin{remark}
\label{rem:regularized-optimizer-existence}
Under~\eqref{eq:XiReachability}, for every \(\eps>0\), the infimum
in~\eqref{eq:opt-regularized} is attained by a unique strategy
\(Q^{*,\eps}\in\cQ_{\mathrm{ac}}\). Indeed, identifying an absolutely
continuous strategy with its trading rate, the feasible set is a nonempty
closed affine subset of \(L^2(\Omega\times[0,T],\cR,\mathbb P\otimes dt)\), while
\(J_\eps\) is weakly lower semicontinuous, coercive, and strictly convex.
\end{remark}

We equip square-integrable c\`adl\`ag processes on $[0,T]$ with the norm
\begin{equation}\label{eq:norm}
  \|X\|_*^2:=\E\left[\int_0^T X_t^2\,dt+X_T^2\right]
\end{equation}
and set
\begin{equation}\label{eq:defCH}
  C_{\cH}^2
  :=\frac12\max\left\{
    \left\|\frac{2\beta+\dot\gamma}{\lambda}\right\|_{L^\infty(0,T)},
    \frac{1}{\lambda_T}
  \right\},
\end{equation}
so that
\begin{equation}\label{eq:CHbound}
  \|X\|_{\cH}\le C_{\cH}\|X\|_*
\end{equation}
for every square-integrable c\`adl\`ag process \(X\), where we recall that $\|\cdot\|_{\cH}$ is the norm in~\eqref{eq:H-norm}. 

We start by recording the following simple Lipschitz property.  

\begin{lemma}\label{le:Lipschitz}
Let \(Q,\widetilde Q\in\cQ\), and let \(Y^Q,Y^{\widetilde Q}\) be the associated
impact processes following~\eqref{eq:OWreversion}. Then
\begin{equation}\label{eq:impact-Lipschitz}
  \|Y^Q-Y^{\widetilde Q}\|_*
  \le L_Y \|Q-\widetilde Q\|_*,
\end{equation}
where
\begin{equation} \label{eq:defLY}
  L_Y:=\|\lambda\|_{L^\infty(0,T)}
       \sqrt{2+T(T+2)\|\beta+\dot\gamma\|_{L^\infty(0,T)}^2}.
\end{equation}
In particular, the map \(Q\mapsto Y^Q\) is affine and Lipschitz on
\((\cQ,\|\cdot\|_*)\).
\end{lemma}

\begin{proof}
Let
\[
  R_t:=Q_t-\widetilde Q_t,
  \qquad
  Z_t:=Y^Q_t-Y^{\widetilde Q}_t,
  \qquad
  B_t:=\int_0^t \beta_s\,ds, \qquad 0 \leq t \leq T,
\]
and introduce the constants
\[
  L_\lambda:=\|\lambda\|_{L^\infty(0,T)},
  \qquad
  L_b:=\|\beta+\dot\gamma\|_{L^\infty(0,T)}.
\]
Then \(R_{0-}=0\), \(Z_{0-}=0\), and \(Z = (Z_t)_{0 \leq t \leq T}\) satisfies
\[
  dZ_t=-\beta_t Z_t\,dt+\lambda_t\,dR_t, \qquad 0 \leq t \leq T,
\]
with explicit solution
\[
  Z_t=e^{-B_t}\int_0^t e^{B_s}\lambda_s\,dR_s \qquad 0 \leq t \leq T.
\]
Since \(e^B\lambda\) is continuous and of finite variation, integration by parts
yields
\[
  Z_t
  = \lambda_t R_t
    -\int_0^t e^{-(B_t-B_s)}\lambda_s(\beta_s+\dot\gamma_s)R_s\,ds.
\]
As \(\beta\ge0\), the exponential factor is bounded by one, and we get the upper bound
\[
  |Z_t|^2
  \le
  2L_\lambda^2 |R_t|^2
  +2L_\lambda^2L_b^2\left(\int_0^t |R_s|\,ds\right)^2.
\]
Next, using Cauchy–Schwarz and Fubini yields
\begin{align*}
  \int_0^T \left(\int_0^t |R_s|\,ds\right)^2dt
  &\le
  \int_0^T t\int_0^t R_s^2\,ds\,dt \\
  &=
  \frac12\int_0^T (T^2-s^2)R_s^2\,ds
  \le
  \frac{T^2}{2}\int_0^T R_s^2\,ds.
\end{align*}
Therefore
\[
  \E\int_0^T Z_t^2\,dt
  \le
  L_\lambda^2\bigl(2+T^2L_b^2\bigr)
  \E\int_0^T R_t^2\,dt.
\]
Likewise,
\[
  |Z_T|^2
  \le
  2L_\lambda^2|R_T|^2
  +2L_\lambda^2L_b^2T\int_0^T R_s^2\,ds.
\]
Adding the two bounds yields 
\eqref{eq:impact-Lipschitz}.
\end{proof}

The next lemma shows that the difference $V(\varepsilon) - V(0)$ can be controlled by the value of the constrained tracking problem of the optimal unregularized inventory $Q^0 \in \cQ$ by an absolutely continuous strategy $Q \in \cQ_{\mathrm{ac}}$ satisfying $Q_T = \Xi_T$.

\begin{lemma}\label{le:rate}
Let
\begin{equation}\label{eq:defC0}
  C_J := C_{\cH}^2L_Y^2.
\end{equation}
Then, for every \(Q\in\cQ\),
\begin{equation}\label{eq:J0-quadratic-growth}
  0\le J_0(Q)-J_0(Q^0)
  = \|Y^Q-Y^0\|_{\cH}^2
  \le C_J \,\E\int_0^T (Q_t-Q_t^0)^2\,dt.
\end{equation}
Consequently, for every \(Q\in\cQ_{\mathrm{ac}}\) and every \(\eps>0\),
\begin{equation}\label{eq:Veps-upper-general}
  0\le V(\eps)-V(0)
  \le
  C_J \,\E\int_0^T (Q_t-Q_t^0)^2\,dt
  +
  \eps\,\E\int_0^T \dot Q_t^2\,dt.
\end{equation}
\end{lemma}

\begin{proof}
Let \(Q\in\cQ\). Since \(Q_T=\Xi_T=Q_T^0\), the constraint identity
\eqref{eq:terminalConstraintNew} applied to \(Y^Q\) and \(Y^0\) shows that
\[
  G(Y^Q-Y^0)=0,
\]
that is, \(Y^Q-Y^0\in\cJ\). Therefore, the first-order condition
\eqref{eq:FOC} gives
\[
  \langle Y^0,Y^Q-Y^0\rangle_{\cH}=0.
\]
It follows that
\[
  J_0(Q)-J_0(Q^0)
  =
  \|Y^Q\|_{\cH}^2-\|Y^0\|_{\cH}^2
  =
  \|Y^Q-Y^0\|_{\cH}^2.
\]
Moreover, using \eqref{eq:CHbound}, \cref{le:Lipschitz}, and again
\(Q_T=Q_T^0\), we obtain
\[
  \|Y^Q-Y^0\|_{\cH}^2
  \le
  C_{\cH}^2\|Y^Q-Y^0\|_*^2
  \le
  C_J\|Q-Q^0\|_*^2
  =
  C_J\,\E\int_0^T (Q_t-Q_t^0)^2\,dt,
\]
which proves \eqref{eq:J0-quadratic-growth} for every \(Q\in\cQ\).

Finally, if \(Q\in\cQ_{\mathrm{ac}}\), then
\[
  V(\eps)-V(0)
  \le
  J_\eps(Q)-J_0(Q^0)
  =
  J_0(Q)-J_0(Q^0)
  +
  \eps\,\E\int_0^T \dot Q_t^2\,dt.
\]
Combining this with \eqref{eq:J0-quadratic-growth} yields
\eqref{eq:Veps-upper-general}.
\end{proof}

We now collect some properties of the unregularized optimal strategy $Q^0$ from \cref{th:OWsolution} to verify that it fits our framework developed in~\cref{sec:constrained}. Recall the definition of the processes $\eta = (\eta_t)_{0 \leq t \leq T}$ and $\theta = (\theta_t)_{0 \leq t \leq T}$ in~\eqref{eq:eta-theta}, as well as the $L^2$ time-translation modulus and associated seminorm from~\eqref{eq:omega} and~\eqref{eq:seminorm}.  

\begin{lemma}\label{le:q0-estimates}
Define
\[
  \alpha_t:=\frac{\eta_t}{\lambda_t},
  \qquad t\in[0,T).
\]
Then \(\alpha\) has finite variation on \([0,T)\), and we set
\[
  A_\alpha:=\|\alpha\|_{L^\infty(0,T)},
  \qquad
  V_\alpha:=|\alpha|_{T-}, %
  \qquad
  B_{\theta}:=\|\dot\theta\|_{L^\infty(0,T)},
  \qquad
  m_2:=\E[M_T^2],
\]
where $|\alpha|_{T-}$ denotes the total variation of $\alpha$ on $[0,T)$. Moreover, defining the constant
\[
  C_{0}:=2A_\alpha^2m_2+\frac{2y^2}{\lambda_0^2},
\]
we have
\begin{equation}\label{eq:q0-initial-bound}
  \E[(Q_0^0)^2]\le C_{0}.
\end{equation}
The $L^2$ time-translation modulus of $Q^0$ satisfies
\begin{equation}\label{eq:q0-omega-bound}
  \omega_{Q^0}(h)
  \le
  3m_2\bigl(A_\alpha^2+4V_\alpha^2+T^2B_\theta^2\bigr)h,
  \qquad 0<h\le T,
\end{equation}
and 
\begin{equation}\label{eq:q0-star-bound}
  [Q^0]^2_{\rm{tr}} \le 3m_2\bigl(A_\alpha^2+4V_\alpha^2+T^2B_\theta^2\bigr).
\end{equation}
Finally,
\begin{equation}\label{eq:q0-K0-bound}
  K_0:=\sup_{t\in[0,T)} \E|Q_t^0-\Xi_t|^2
  \le
  6m_2\bigl(A_\alpha^2+T^2B_\theta^2\bigr)
  +\frac{6y^2}{\lambda_0^2}
  +2\E[\Xi_T^2].
\end{equation}
\end{lemma}

\begin{proof}
By~\cref{ass:beta-lambda}, \(\eta\) has finite variation on \([0,T)\) and \(\lambda^{-1}\) is
absolutely continuous. Hence, \(\alpha=\eta/\lambda\) has finite variation on
\([0,T)\) and we have $A_\alpha < \infty$, as well as $V_\alpha < \infty$. Moreover, $B_\theta < \infty$, again by~\cref{ass:beta-lambda}, and $m_2 < \infty$ by~\cref{pr:minimizerH}.

Next, recall the representation of \(Q^0\) from~\eqref{eq:optQ} given by
\begin{equation} \label{eq:Q0}
  Q_t^0 = \alpha_t M_t - \frac{y}{\lambda_0} + \int_0^t \dot\theta_s M_s\,ds, \qquad 0 \leq t < T.
\end{equation}
At \(t=0\), we have
\[
  Q_0^0=\alpha_0M_0-\frac{y}{\lambda_0},
\]
and therefore
\[
  \E[(Q_0^0)^2]
  \le
  2A_\alpha^2\E[M_0^2]+\frac{2y^2}{\lambda_0^2}
  \le
  2A_\alpha^2m_2+\frac{2y^2}{\lambda_0^2},
\]
which proves \eqref{eq:q0-initial-bound}.

Fix \(0<h\le T\) and write \(s:=(t-h)^+\). For \(t<T\), using~\eqref{eq:Q0} gives
\[
  Q_t^0-Q_s^0
  =
  \alpha_t(M_t-M_s)+(\alpha_t-\alpha_s)M_s+\int_s^t \dot\theta_r M_r\,dr.
\]
Hence, the estimate \((a+b+c)^2\le 3(a^2+b^2+c^2)\) yields
\begin{align*}
  \omega_{Q^0}(h) ={} & \E\int_0^T \left|Q^0_t-Q^0_{s}\right|^2\,dt
  \\ \leq{} &
  3A_\alpha^2\,\E\int_0^T |M_t-M_s|^2\,dt +3\,\E\!\left[\sup_{u\in[0,T)} M_u^2\right]
    \int_0^T |\alpha_t-\alpha_s|^2\,dt \\
  &+3B_\theta^2\,\E\int_0^T \left(\int_s^t |M_r|\,dr\right)^2dt.
\end{align*}
For the martingale term, due to the computations in~\eqref{eq:martingale-increment}, we have
\[
  \E\int_0^T |M_t-M_s|^2\,dt
  \le h\,\E[\langle M\rangle_T]
  \le h\,m_2.
\]
For the finite variation term, performing the same computations as in~\eqref{eq:finite-variation-increment} gives
\begin{align*}
  \int_0^T |\alpha_t-\alpha_s|^2\,dt \le hV_\alpha^2.
\end{align*}
Moreover, using Doob's inequality, we have
\[
  \E\!\left[\sup_{u\in[0,T)}M_u^2\right]\le 4m_2.
\]
Finally, Cauchy–Schwarz and Fubini imply, again similar to~\eqref{eq:finite-variation-increment}, the bound
\begin{align*}
  \int_0^T \left(\int_s^t |M_r|\,dr\right)^2dt
  &\le
  h\int_0^T\int_s^t M_r^2\,dr\,dt \le h^2\int_0^T M_r^2\,dr.
\end{align*}
Taking expectations and using \(h\le T\) as well as \(\E[M_r^2]\le m_2\), we get
\[
  \E\int_0^T \left(\int_s^t |M_r|\,dr\right)^2dt
  \le
  T^2hm_2.
\]
Combining the three estimates proves \eqref{eq:q0-omega-bound}. Since
\(\omega_{Q^0}(h)=\omega_{Q^0}(T)\) for every \(h\ge T\),
\eqref{eq:q0-star-bound} follows immediately.

For the claim in~\eqref{eq:q0-K0-bound}, using once more the representation in~\eqref{eq:Q0} as well as \((a+b+c)^2\le 3(a^2+b^2+c^2)\) gives
\[
  |Q_t^0|^2
  \le
  3A_\alpha^2M_t^2+\frac{3y^2}{\lambda_0^2}
  +3B_\theta^2T\int_0^T M_s^2\,ds,
  \qquad t<T.
\]
Taking expectations and using \(\E[M_t^2]\le m_2\) yields
\[
  \sup_{t\in[0,T)}\E|Q_t^0|^2
  \le
  3m_2\bigl(A_\alpha^2+T^2B_\theta^2\bigr)+\frac{3y^2}{\lambda_0^2}.
\]
Lastly, since \(\sup_{t\in[0,T]}\E[\Xi_t^2]\le \E[\Xi_T^2]\), we conclude that
\[
  K_0
  \le
  2\sup_{t\in[0,T)}\E|Q_t^0|^2
  +2\sup_{t\in[0,T]}\E[\Xi_t^2],
\]
which is \eqref{eq:q0-K0-bound}.
\end{proof}

The next result exhibits an absolutely continuous execution strategy $Q^\eps$ which is based on the unregularized optimizer $Q^0$, whose explicit form was derived in \cref{th:OWsolution}. Following the construction in the proof of \cref{thm:constrained-master}, $Q^\eps$
tracks $Q^0$ through an exponential filter with
relaxation scale $\sqrt{\eps}$ and, over the final interval of that length,
switches to a terminal bridge enforcing $Q_T^\eps=\Xi_T$. The resulting
feedback rule is explicit and does not require solving the regularized stochastic control problem.

\begin{lemma}[Explicit execution strategy]\label{le:tracking}
Fix \(0<\eps\le T^2/4\). Define
\begin{equation}\label{eq:Qeps-filter}
  \bar Q_t^\eps
  :=\frac{1}{\sqrt{\eps}}\int_0^t
  e^{-(t-r)/\sqrt{\eps}}Q_r^0\,dr,
  \qquad 0\le t<T,
\end{equation}
and
\begin{equation}\label{eq:Qeps-explicit}
  Q_t^\eps
  :=
  \begin{cases}
    \displaystyle
    \bar Q_t^\eps,
    & 0\le t\le T-\sqrt{\eps},\\[2.2ex]
    \displaystyle
    \frac{T-t}{\sqrt{\eps}}\,
    \bar Q_{T-\sqrt{\eps}}^\eps
    +(T-t)\int_{T-\sqrt{\eps}}^t
    \frac{\Xi_r}{(T-r)^2}\,dr,
    & T-\sqrt{\eps}<t<T,\\[2.2ex]
    \Xi_T,
    & t=T.
  \end{cases}
\end{equation}
Then \(Q^\eps\in\cQ_{\mathrm{ac}}\). Its trading rate is given by the
feedback rule
\begin{equation}\label{eq:Qeps-rate}
  \dot Q_t^\eps
  =
  \begin{cases}
    \displaystyle
    \frac{Q_t^0-Q_t^\eps}{\sqrt{\eps}},
    & 0\le t<T-\sqrt{\eps},\\[2ex]
    \displaystyle
    \frac{\Xi_t-Q_t^\eps}{T-t},
    & T-\sqrt{\eps}\le t<T,
  \end{cases}
  \qquad \mathbb P\otimes dt\text{-a.e.}
\end{equation}

Define
\begin{equation}\label{eq:defCtrTilde}
  \widetilde C_{\mathrm{tr}}
  :=2\left(
    C_{0}+2[Q^0]^2_{\rm{tr}}+8K_0+10\E[\Xi_T^2]
  \right)
\end{equation}
and
\begin{equation}\label{eq:defCtr}
  C_{\mathrm{tr}}
  :=\widetilde C_{\mathrm{tr}}+2T\Lambda_{\Xi}.
\end{equation}
Then
\begin{equation}\label{eq:tracking}
  \E\left[
    \int_0^T (Q_t^\eps-Q_t^0)^2\,dt
    +
    \eps\int_0^T (\dot Q_t^\eps)^2\,dt
  \right]
  \le
  \widetilde C_{\mathrm{tr}}\sqrt{\eps}
  +4\Lambda_{\Xi}\eps
  \le
  C_{\mathrm{tr}}\sqrt{\eps}.
\end{equation}
\end{lemma}

\begin{proof}
For brevity, set $a:=\sqrt{\eps}$ and $s:=T-\sqrt{\eps}$. On \([0,s]\), differentiating
\eqref{eq:Qeps-filter} gives
\[
  \dot Q_t^\eps=\frac{Q_t^0-Q_t^\eps}{a}.
\]
Thus, the calculation in the proof of \cref{thm:master}, applied with target
\(Q^0\) and initial position zero, together with
\eqref{eq:q0-initial-bound} and \eqref{eq:q0-star-bound}, yields
\begin{equation}\label{eq:Qeps-first-phase}
  \E\left[
    \frac12\int_0^s (Q_t^\eps-Q_t^0)^2\,dt
    +\frac{a^2}{2}\int_0^s (\dot Q_t^\eps)^2\,dt
  \right]
  \le a\left(C_{0}+2[Q^0]^2_{\rm{tr}}\right).
\end{equation}

Set
\[
  B_\eps:=\E\bigl[|\Xi_s-\bar Q_s^\eps|^2\bigr].
\]
By Jensen's inequality,
\begin{align*}
  \E\bigl[|\bar Q_s^\eps|^2\bigr]
  &\le \frac1a\int_0^s e^{-(s-r)/a}\E[|Q_r^0|^2] \,dr
  \le \sup_{t\in[0,T)}\E[|Q_t^0|^2]
  \le 2K_0+2\E[\Xi_T^2],
\end{align*}
where the last inequality follows from
\(|Q_t^0|^2\le2|Q_t^0-\Xi_t|^2+2|\Xi_t|^2\). Consequently,
\begin{equation}\label{eq:Qeps-B-bound}
  B_\eps
  \le 2\E[\Xi_s^2]+2\E[|\bar Q_s^\eps|^2]
  \le 4K_0+6\E[\Xi_T^2].
\end{equation}

On \([s,T)\), the second branch of \eqref{eq:Qeps-explicit} solves the
second feedback equation in \eqref{eq:Qeps-rate}. Integration by parts gives
\begin{equation}\label{eq:Qeps-terminal-rate}
  \dot Q_t^\eps
  =
  \frac{\Xi_s-\bar Q_s^\eps}{a}
  +N_t,
  \qquad
  N_t:=\int_{(s,t]}\frac{1}{T-r}\,d\Xi_r,
  \qquad s\le t<T.
\end{equation}
The process \(N\) is progressively measurable, and the martingale isometry and
Fubini's theorem yield
\begin{equation}\label{eq:Qeps-N-bound}
  \E\int_s^T N_t^2\,dt
  =\int_{(s,T]}\frac{d\E[\Xi_r^2]}{T-r}
  \le\Lambda_{\Xi}.
\end{equation}
Together with \eqref{eq:Qeps-first-phase} and \eqref{eq:Qeps-B-bound}, this
shows that the trading rate in \eqref{eq:Qeps-rate} is progressively measurable and
square-integrable. Moreover, the argument following
\eqref{eq:proof-terminal-state2} shows that
\(Q_t^\eps\to\Xi_T\) almost surely as \(t\uparrow T\). Hence
\(Q^\eps\in\cQ_{\mathrm{ac}}\).

For \(s\le t<T\), \eqref{eq:Qeps-rate} and
\eqref{eq:Qeps-terminal-rate} imply
\begin{align*}
  \E[|\Xi_t-Q_t^\eps|^2]
  &\le
  2\frac{(T-t)^2}{a^2}B_\eps
  +2(T-t)\Lambda_{\Xi}.
\end{align*}
Indeed, for \(r\le t\),
\((T-t)^2/(T-r)^2\le (T-t)/(T-r)\), and the second term follows from the
martingale isometry. Integrating over \([s,T]\) and using the definition of
\(K_0\), we obtain
\begin{equation}\label{eq:Qeps-terminal-tracking}
  \frac12\E\int_s^T (Q_t^\eps-Q_t^0)^2\,dt
  \le \frac{2a}{3}B_\eps+a^2\Lambda_{\Xi}+aK_0.
\end{equation}
Similarly, \eqref{eq:Qeps-terminal-rate} and
\eqref{eq:Qeps-N-bound} give
\begin{equation}\label{eq:Qeps-terminal-trading}
  \frac{a^2}{2}\E\int_s^T (\dot Q_t^\eps)^2\,dt
  \le aB_\eps+a^2\Lambda_{\Xi}.
\end{equation}
Combining \eqref{eq:Qeps-terminal-tracking} and
\eqref{eq:Qeps-terminal-trading} with \eqref{eq:Qeps-B-bound} yields
\begin{align*}
  \E\left[
    \frac12\int_s^T (Q_t^\eps-Q_t^0)^2\,dt
    +\frac{a^2}{2}\int_s^T (\dot Q_t^\eps)^2\,dt
  \right]
  &\le
  \frac{5a}{3}B_\eps+aK_0+2a^2\Lambda_{\Xi}\\
  & \le
  a\left(8K_0+10\E[\Xi_T^2]\right)
  +2a^2\Lambda_{\Xi}.
\end{align*}
Combining this estimate with \eqref{eq:Qeps-first-phase} and multiplying by
two gives
\[
  \E\left[
    \int_0^T (Q_t^\eps-Q_t^0)^2\,dt
    +a^2\int_0^T (\dot Q_t^\eps)^2\,dt
  \right]
  \le
  \widetilde C_{\mathrm{tr}}a+4a^2\Lambda_{\Xi}.
\]
Since \(a=\sqrt\eps\) and \(a\le T/2\), this proves
\eqref{eq:tracking}. The finiteness of the constants follows from
\cref{le:q0-estimates} and \eqref{eq:XiReachability}.
\end{proof}

\begin{remark}[On \(Q^\eps\)]
The above construction of \(Q^\eps\) deviates slightly from the proof of
\cref{thm:constrained-master}. There, under the general assumptions of the
constrained tracking problem, the switching time is selected by an averaging
argument from the interval
\([T-2\sqrt{\eps},T-\sqrt{\eps}]\). Here, the additional uniform estimate
\(
  \sup_{t<T}\E|Q_t^0-\Xi_t|^2<\infty
\)
from \cref{le:q0-estimates} allows us to use the deterministic
switching time \(T-\sqrt{\eps}\) and hence make the resulting strategy
\eqref{eq:Qeps-explicit} fully explicit.

An example is shown in Figure~\ref{fig:OW-strategies}. We remark that the slope of the terminal bridge could be changed by modifying the switching time to \(T-c\sqrt{\eps}\) for some $c>0$.
\end{remark}

We can now state our main theorem. Its bounds in~\eqref{eq:OW-subopt} and~\eqref{eq:OW-rate-competitor} apply to the explicit strategy $Q^\eps$ from~\eqref{eq:Qeps-explicit}, not merely to the optimal value, and show that the excess impact cost $J_0(Q^\eps)-J_0(Q^0)$ and the total regularized cost above $V(0)$ are both $O(\sqrt{\eps})$.

\begin{theorem}[Main result]\label{th:OW-rate}
For every \(0<\eps\le T^2/4\), the explicit strategy \(Q^\eps\in\cQ_{\mathrm{ac}}\) from \eqref{eq:Qeps-explicit} satisfies
\begin{align}
  0\le J_0(Q^\eps)-J_0(Q^0)
  &\le
  C_J\bigl(\widetilde C_{\mathrm{tr}}\sqrt{\eps}+4\Lambda_{\Xi}\eps\bigr),
  \label{eq:OW-subopt}\\
  0\le J_\eps(Q^\eps)-V(0)
  &\le
  (1+C_J)\bigl(\widetilde C_{\mathrm{tr}}\sqrt{\eps}+4\Lambda_{\Xi}\eps\bigr).
  \label{eq:OW-rate-competitor}
\end{align}
As a consequence, the value functions satisfy
\begin{equation}\label{eq:OW-rate}
  0\le V(\eps)-V(0)
  \le
  (1+C_J)\bigl(\widetilde C_{\mathrm{tr}}\sqrt{\eps}+4\Lambda_{\Xi}\eps\bigr)
\end{equation}
and in particular
\[
  V(\eps)-V(0)=O(\sqrt{\eps}),
  \qquad
  \eps\downarrow0.
\]
\end{theorem}

\begin{proof}
Let \(Q^\eps\) be the explicit strategy from \eqref{eq:Qeps-explicit}. Since
\[
  \E\int_0^T (Q_t^\eps-Q_t^0)^2\,dt
  \le
  \E\left[
    \int_0^T (Q_t^\eps-Q_t^0)^2\,dt
    +
    \eps\int_0^T (\dot Q_t^\eps)^2\,dt
  \right],
\]
\cref{le:rate,le:tracking} give
\[
  0\le J_0(Q^\eps)-J_0(Q^0)
  \le
  C_J\bigl(\widetilde C_{\mathrm{tr}}\sqrt{\eps}+4\Lambda_{\Xi}\eps\bigr),
\]
which is \eqref{eq:OW-subopt}. Adding the regularization term and noting that likewise
\[
  \eps \, \E\left[
    \int_0^T (\dot Q_t^\eps)^2\,dt
  \right]
  \le
  \E\left[
    \int_0^T (Q_t^\eps-Q_t^0)^2\,dt
    +
    \eps\int_0^T (\dot Q_t^\eps)^2\,dt
  \right]
\]
yields
\eqref{eq:OW-rate-competitor}. Since \(V(\eps)\le J_\eps(Q^\eps)\),
\eqref{eq:OW-rate} follows.
\end{proof}

Recall that convergence of $Q^\eps$ to $Q^0$  in $L^2$ was already shown in~\eqref{eq:tracking}. The following complements that result with the convergence of the true optimizer~$Q^{*,\eps}$. 

\begin{corollary}[Convergence of the regularized optimizer]
\label{cor:regularized-optimizer}
Let \(Q^{*,\eps}\) be the regularized optimizer from
\cref{rem:regularized-optimizer-existence}, and let
\(Y^{*,\eps}:=Y^{Q^{*,\eps}}\). Define
\[
  \underline c_{\cH}
  :=
  \inf_{t\in[0,T]}
  \frac{2\beta_t+\dot\gamma_t}{2\lambda_t}
  >0
\qquad\text{and}\qquad
  C_{\mathrm{inv}}
  :=
  \frac{
    2\|\lambda^{-1}\|_{L^\infty(0,T)}^2
    +
    T^2
    \left\|
      \frac{\beta+\dot\gamma}{\lambda}
    \right\|_{L^\infty(0,T)}^2
  }{
    \underline c_{\cH}
  }.
\]
Then, for every \(0<\eps\le T^2/4\),
\begin{equation}\label{eq:regularized-optimizer-impact}
  0
  \le
  J_0(Q^{*,\eps})-J_0(Q^0)
  =
  \|Y^{*,\eps}-Y^0\|_{\cH}^2
  \le
  (1+C_J)
  \bigl(
    \widetilde C_{\mathrm{tr}}\sqrt{\eps}
    +4\Lambda_\Xi\eps
  \bigr),
\end{equation}
and
\begin{equation}\label{eq:regularized-optimizer-inventory}
  \E\int_0^T
  |Q_t^{*,\eps}-Q_t^0|^2\,dt
  \le
  C_{\mathrm{inv}}(1+C_J)
  \bigl(
    \widetilde C_{\mathrm{tr}}\sqrt{\eps}
    +4\Lambda_\Xi\eps
  \bigr).
\end{equation}
In particular,
\[
  J_0(Q^{*,\eps})-J_0(Q^0)
  =
  O(\sqrt{\eps}),
  \qquad
  \E\int_0^T
  |Q_t^{*,\eps}-Q_t^0|^2\,dt
  =
  O(\sqrt{\eps}),
  \qquad
  \eps\downarrow0.
\]
\end{corollary}

\begin{proof}
Since \(Q^{*,\eps}\) minimizes \(J_\eps\) over
\(\cQ_{\mathrm{ac}}\), we have
\[
  J_0(Q^{*,\eps})-J_0(Q^0)
  \le
  J_\eps(Q^{*,\eps})-J_0(Q^0)
  =
  V(\eps)-V(0).
\]
The identity in~\eqref{eq:J0-quadratic-growth} and the value bound
\eqref{eq:OW-rate} therefore yield
\eqref{eq:regularized-optimizer-impact}.

It remains to control the inventory difference. More generally, let
\(Q\in\cQ\), and set
\[
  R_t:=Q_t-Q_t^0,
  \qquad
  Z_t:=Y_t^Q-Y_t^0.
\]
Subtracting the two instances of~\eqref{eq:Qexplicit} gives
\[
  R_t
  =
  \frac{Z_t}{\lambda_t}
  +
  \int_0^t
  \frac{\beta_s+\dot\gamma_s}{\lambda_s}Z_s\,ds,
  \qquad 0\le t\le T.
\]
Consequently, by Cauchy--Schwarz and Fubini,
\begin{align*}
  \E\int_0^T R_t^2\,dt
  &\le
  2\|\lambda^{-1}\|_{L^\infty(0,T)}^2
  \E\int_0^T Z_t^2\,dt\\
  &\quad
  +
  2
  \left\|
    \frac{\beta+\dot\gamma}{\lambda}
  \right\|_{L^\infty(0,T)}^2
  \E\int_0^T
  t\int_0^t Z_s^2\,ds\,dt\\
  &\le
  \left(
    2\|\lambda^{-1}\|_{L^\infty(0,T)}^2
    +
    T^2
    \left\|
      \frac{\beta+\dot\gamma}{\lambda}
    \right\|_{L^\infty(0,T)}^2
  \right)
  \E\int_0^T Z_t^2\,dt.
\end{align*}
Moreover, by the definition of the \(\cH\)-norm,
\[
  \E\int_0^T Z_t^2\,dt
  \le
  \frac{1}{\underline c_{\cH}}
  \|Z\|_{\cH}^2.
\]
Thus,
\[
  \E\int_0^T (Q_t-Q_t^0)^2\,dt
  \le
  C_{\mathrm{inv}}
  \|Y^Q-Y^0\|_{\cH}^2.
\]
Applying this estimate to \(Q=Q^{*,\eps}\) and using
\eqref{eq:regularized-optimizer-impact} proves
\eqref{eq:regularized-optimizer-inventory}.
\end{proof}

Our final result shows that the generic rate $\sqrt{\varepsilon}$ is sharp in the benchmark case where $\beta, \lambda > 0$ are constants, the terminal target position $\Xi_T$ is deterministic, and $y=0$. 

\begin{proposition}[Sharpness]\label{pr:OW-rate-sharp}
Suppose $y=0$, that \(\beta,\lambda>0\) are constants, and that the target position $\Xi_T$ is deterministic. Then, 
\[
V(\varepsilon)
=
\frac{\lambda \Xi_T^2}{\beta T + 2}
+
\frac{2\sqrt{\beta\lambda}\,\Xi_T^2}
{(\beta T+2)^2}\sqrt{\varepsilon}
+
O(\varepsilon),
\qquad \varepsilon\downarrow0.
\]
Consequently, if \(\Xi_T\neq0\), then
\[
V(\varepsilon)-V(0)
\sim
\frac{2\sqrt{\beta\lambda}\,\Xi_T^2}
{(\beta T+2)^2}\sqrt{\varepsilon}, \qquad \varepsilon\downarrow0, 
\]
and hence the rate \(O(\sqrt{\varepsilon})\) is sharp.
\end{proposition}
\begin{proof}
Set \(D:=\beta T+2\). We first consider the unregularized case
\(\varepsilon=0\). Under the stated assumptions, this is the original
Obizhaeva--Wang model~\cite{ObizhaevaWang.13}. By
\cite[Proposition~2]{ObizhaevaWang.13}, the optimal strategy is
\[
  dQ_t^0
  =
  \frac{\Xi_T}{D}\,\delta_0(dt)
  +\frac{\beta\Xi_T}{D}\,dt
  +\frac{\Xi_T}{D}\,\delta_T(dt),
  \qquad 0\le t\le T,
\]
and its cost is
\begin{equation}\label{eq:value-unreg}
  V(0)=\frac{\lambda\Xi_T^2}{D}.
\end{equation}

We next use the explicit solution of the regularized problem
\eqref{eq:opt-regularized} obtained by Chen, Horst, and
Tran~\cite[Theorem~2.1]{ChenHorstTran.25}. Define
\[
  k_\varepsilon
  :=
  \sqrt{\beta^2+\frac{\beta\lambda}{\varepsilon}},
  \qquad
  K_\varepsilon:=\frac{k_\varepsilon T}{2},
\]
and
\[
  a_\varepsilon
  :=
  \frac{\lambda}{\varepsilon}\sinh(K_\varepsilon),
  \qquad
  b_\varepsilon
  :=
  \beta k_\varepsilon\cosh(K_\varepsilon)
  +k_\varepsilon^2\sinh(K_\varepsilon),
  \qquad
  c_\varepsilon:=\frac{\lambda}{\varepsilon}.
\]
Then the optimal inventory and trading rate are
\[
  Q_t^{*,\varepsilon}
  =
  \Xi_T
  \frac{
    a_\varepsilon+b_\varepsilon t
    +c_\varepsilon\sinh\!\left(k_\varepsilon(t-T/2)\right)
  }{
    2a_\varepsilon+b_\varepsilon T
  },
  \qquad 0\le t\le T,
\]
and
\[
  u_t^{*,\varepsilon}
  =
  d_\varepsilon
  +
  \frac{\Xi_Tc_\varepsilon k_\varepsilon}
       {2a_\varepsilon+b_\varepsilon T}
  \cosh\!\left(k_\varepsilon(t-T/2)\right),
  \qquad 0\le t<T,
\]
where
\begin{equation}\label{def:d-eps}
  d_\varepsilon
  :=
  \frac{\Xi_Tb_\varepsilon}{2a_\varepsilon+b_\varepsilon T}.
\end{equation}

The optimal value is not stated explicitly in
\cite{ChenHorstTran.25}, but it follows from their proof. Indeed,
\cite[Appendix~6.1]{ChenHorstTran.25} gives the first-order condition
\[
  2\varepsilon u_t^{*,\varepsilon}
  +
  \lambda\int_0^T e^{-\beta|t-s|}
  u_s^{*,\varepsilon}\,ds
  +
  \mu_\varepsilon
  =
  0
\]
for some Lagrange multiplier \(\mu_\varepsilon\), together with
\[
  \mu_\varepsilon
  =
  -\frac{2\varepsilon k_\varepsilon^2}{\beta^2}d_\varepsilon.
\]
For an absolutely continuous strategy with trading rate \(u\), the objective
in the present constant-coefficient setting can be written as
\[
  J_\varepsilon(Q)
  =
  \E\left[
    \frac{\lambda}{2}
    \int_0^T\int_0^T
    e^{-\beta|t-s|}u_tu_s\,ds\,dt
    +
    \varepsilon\int_0^T u_t^2\,dt
  \right].
\]
Multiplying the first-order condition by \(u_t^{*,\varepsilon}\), integrating
over \([0,T]\), and using
\(\int_0^T u_t^{*,\varepsilon}\,dt=\Xi_T\), we therefore obtain
\begin{equation}\label{eq:value-reg}
  V(\varepsilon)
  =
  J_\varepsilon(Q^{*,\varepsilon})
  =
  -\frac{\mu_\varepsilon\Xi_T}{2}
  =
  \frac{\varepsilon k_\varepsilon^2}{\beta^2}
  d_\varepsilon\Xi_T.
\end{equation}

We now expand \(d_\varepsilon\). From \eqref{def:d-eps},
\begin{equation}\label{def:d-eps2}
  d_\varepsilon
  =
  \Xi_T
  \frac{
    k_\varepsilon^2
    +\beta k_\varepsilon\coth(K_\varepsilon)
  }{
    T\bigl(k_\varepsilon^2
    +\beta k_\varepsilon\coth(K_\varepsilon)\bigr)
    +\dfrac{2}{\beta}(k_\varepsilon^2-\beta^2)
  }.
\end{equation}
Since \(k_\varepsilon\to\infty\) and
\(
  \coth(K_\varepsilon)
  =
  1+O(e^{-k_\varepsilon T}),
\)
we obtain
\[
  k_\varepsilon^2
  +\beta k_\varepsilon\coth(K_\varepsilon)
  =
  k_\varepsilon^2+\beta k_\varepsilon
  +O(k_\varepsilon^{-2}),
\]
and
\[
  T\bigl(k_\varepsilon^2
  +\beta k_\varepsilon\coth(K_\varepsilon)\bigr)
  +\frac{2}{\beta}(k_\varepsilon^2-\beta^2)
  =
  \frac{D}{\beta}k_\varepsilon^2
  +\beta T k_\varepsilon
  -2\beta
  +O(k_\varepsilon^{-2}).
\]
Hence
\[
  d_\varepsilon
  =
  \frac{\beta\Xi_T}{D}
  \frac{
    1+\dfrac{\beta}{k_\varepsilon}
    +O(k_\varepsilon^{-2})
  }{
    1+\dfrac{\beta^2T}{Dk_\varepsilon}
    +O(k_\varepsilon^{-2})
  }.
\]
Using \((1+x)^{-1}=1-x+O(x^2)\) and
\(1-\beta T/D=2/D\), we conclude that
\begin{equation}\label{eq:d-eps-asymp}
  d_\varepsilon
  =
  \frac{\beta\Xi_T}{D}
  +
  \frac{2\beta^2\Xi_T}{D^2}\frac{1}{k_\varepsilon}
  +
  O(k_\varepsilon^{-2}),
  \qquad
  \varepsilon\downarrow0.
\end{equation}

The remaining factor in \eqref{eq:value-reg} satisfies
\begin{equation}\label{eq:k-eps-exact}
  \frac{\varepsilon k_\varepsilon^2}{\beta^2}
  =
  \frac{\lambda}{\beta}+\varepsilon.
\end{equation}
Moreover,
\[
  \frac{1}{k_\varepsilon}
  =
  \sqrt{\frac{\varepsilon}{\beta\lambda}}
  \left(
    1+\frac{\beta\varepsilon}{\lambda}
  \right)^{-1/2}
  =
  \sqrt{\frac{\varepsilon}{\beta\lambda}}
  +O(\varepsilon^{3/2}),
\]
and therefore
\(
  \frac{1}{k_\varepsilon^2}=O(\varepsilon) .%
\)
Substituting \eqref{eq:d-eps-asymp} and
\eqref{eq:k-eps-exact} into \eqref{eq:value-reg} gives
\begin{align*}
  V(\varepsilon)
  &=
  \left(\frac{\lambda}{\beta}+\varepsilon\right)
  \left(
    \frac{\beta\Xi_T}{D}
    +
    \frac{2\beta^2\Xi_T}{D^2}\frac{1}{k_\varepsilon}
    +
    O(k_\varepsilon^{-2})
  \right)\Xi_T\\
  &=
  \frac{\lambda\Xi_T^2}{D}
  +
  \frac{2\lambda\beta\Xi_T^2}{D^2}
  \frac{1}{k_\varepsilon}
  +
  O(\varepsilon)
  =
  \frac{\lambda\Xi_T^2}{D}
  +
  \frac{2\sqrt{\beta\lambda}\,\Xi_T^2}{D^2}
  \sqrt{\varepsilon}
  +
  O(\varepsilon).
\end{align*}
Together with \eqref{eq:value-unreg}, this proves the claimed expansion. If
\(\Xi_T\neq0\), the coefficient of \(\sqrt{\varepsilon}\) is strictly
positive, so the asserted asymptotic equivalence and sharpness follow.
\end{proof}

Figure~\ref{fig:OW-strategies} illustrates the two regularizations in
the constant-coefficient benchmark of \cref{pr:OW-rate-sharp}. As
\(\varepsilon\) decreases, both the simplified strategy \(Q^\varepsilon\) and
the regularized optimizer \(Q^{*,\varepsilon}\) approach the singular
unregularized optimizer \(Q^0\).

\begin{figure}[t]
\centering
\includegraphics[width=0.7\textwidth]{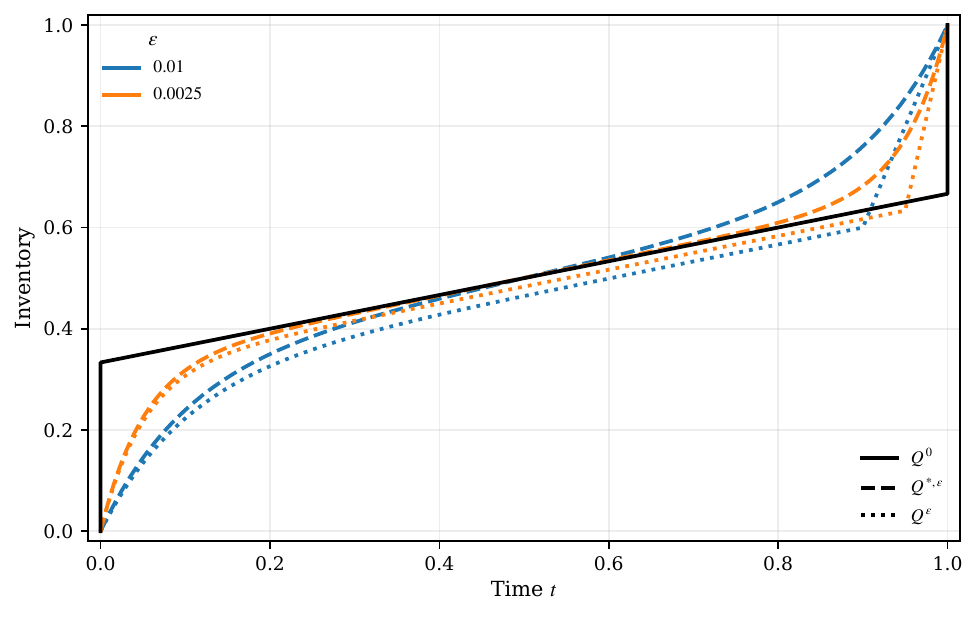}
\caption{Comparison of the unregularized optimal strategy \(Q^0\), the explicit
strategy \(Q^\varepsilon\) from~\eqref{eq:Qeps-explicit}, and the regularized
optimal strategy \(Q^{*,\varepsilon}\) in the constant-coefficient benchmark
with \(T=1\), \(\beta=\lambda=1\), \(y=0\), and deterministic terminal target
\(\Xi_T=1\). The colored dotted and dashed
lines show \(Q^\varepsilon\) and \(Q^{*,\varepsilon}\), respectively, for
\(\varepsilon\in\{0.01,0.0025\}\).}
\label{fig:OW-strategies}
\end{figure}

\bibliographystyle{abbrv}
\bibliography{cite}

@Preamble{"\newcommand{\dummy}[1]{}"}

@book{CarteaJaimungalPenalva.15,
  title={Algorithmic and High-Frequency Trading},
  author={Cartea, {\'A}. and Jaimungal, S. and Penalva, J.},
  year={2015},
  publisher={Cambridge University Press}
}

@book{Webster.23,
  title={Handbook of Price Impact Modeling},
  author={Webster, Kevin},
  year={2023},
  publisher={CRC Press},
  address={Boca Raton, FL},
}

@article{annkirchner.kruse.15,
  AUTHOR   = {Ankirchner, Stefan and Kruse, Thomas},
  TITLE    = {Optimal position targeting with stochastic linear--quadratic costs},
  JOURNAL  = {Banach Center Publ.},
  FJOURNAL = {Banach Center Publications},
  VOLUME   = {104},
  NUMBER   = {1},
  PAGES    = {9--24},
  YEAR     = {2015},
  DOI      = {10.4064/bc104-0-1},
  URL      = {https://doi.org/10.4064/bc104-0-1},
}

@article{BankVoss.18,
  AUTHOR   = {Bank, Peter and Vo{\ss}, Moritz},
  TITLE    = {Linear quadratic stochastic control problems with stochastic terminal constraint},
  JOURNAL  = {SIAM J. Control Optim.},
  FJOURNAL = {SIAM Journal on Control and Optimization},
  VOLUME   = {56},
  NUMBER   = {2},
  PAGES    = {672--699},
  YEAR     = {2018},
  DOI      = {10.1137/16M1104597},
  URL      = {https://doi.org/10.1137/16M1104597},
}

@article{DolinskyGottesmanGurelGurevich.20,
  AUTHOR   = {Dolinsky, Yan and Gottesman, Benjamin and Gurel-Gurevich, Ori},
  TITLE    = {A note on costs minimization with stochastic target constraints},
  JOURNAL  = {Electron. Commun. Probab.},
  FJOURNAL = {Electronic Communications in Probability},
  VOLUME   = {25},
  NUMBER   = {11},
  PAGES    = {1--12},
  YEAR     = {2020},
  DOI      = {10.1214/20-ECP295},
  URL      = {https://doi.org/10.1214/20-ECP295},
}

@article{GraeweHorst.17,
  AUTHOR   = {Graewe, Paulwin and Horst, Ulrich},
  TITLE    = {Optimal trade execution with instantaneous price impact and stochastic resilience},
  JOURNAL  = {SIAM J. Control Optim.},
  FJOURNAL = {SIAM Journal on Control and Optimization},
  VOLUME   = {55},
  NUMBER   = {6},
  PAGES    = {3707--3725},
  YEAR     = {2017},
  DOI      = {10.1137/16M1105463},
  URL      = {https://doi.org/10.1137/16M1105463},
}

@article{bank.al.17,
  AUTHOR   = {Bank, Peter and Soner, H. Mete and Vo{\ss}, Moritz},
  TITLE    = {Hedging with temporary price impact},
  JOURNAL  = {Math. Financ. Econ.},
  FJOURNAL = {Mathematics and Financial Economics},
  VOLUME   = {11},
  NUMBER   = {2},
  PAGES    = {215--239},
  YEAR     = {2017},
  DOI      = {10.1007/s11579-016-0178-4},
  URL      = {https://doi.org/10.1007/s11579-016-0178-4},
}

@article{ObizhaevaWang.13,
  AUTHOR   = {Obizhaeva, Anna A. and Wang, Jiang},
  TITLE    = {Optimal trading strategy and supply/demand dynamics},
  JOURNAL  = {J. Financial Mark.},
  FJOURNAL = {Journal of Financial Markets},
  VOLUME   = {16},
  NUMBER   = {1},
  PAGES    = {1--32},
  YEAR     = {2013},
  DOI      = {10.1016/j.finmar.2012.09.001},
  URL      = {https://doi.org/10.1016/j.finmar.2012.09.001},
}

@article{BankFruth.14,
  AUTHOR   = {Bank, Peter and Fruth, Antje},
  TITLE    = {Optimal order scheduling for deterministic liquidity patterns},
  JOURNAL  = {SIAM J. Financial Math.},
  FJOURNAL = {SIAM Journal on Financial Mathematics},
  VOLUME   = {5},
  NUMBER   = {1},
  PAGES    = {137--152},
  YEAR     = {2014},
  DOI      = {10.1137/120897511},
  URL      = {https://doi.org/10.1137/120897511},
}

@article{AlfonsiAcevedo.14,
  AUTHOR   = {Alfonsi, Aur{\'e}lien and Acevedo, Jos{\'e} Infante},
  TITLE    = {Optimal execution and price manipulations in time-varying limit order books},
  JOURNAL  = {Appl. Math. Finance},
  FJOURNAL = {Applied Mathematical Finance},
  VOLUME   = {21},
  NUMBER   = {3},
  PAGES    = {201--237},
  YEAR     = {2014},
  DOI      = {10.1080/1350486X.2013.845471},
  URL      = {https://doi.org/10.1080/1350486X.2013.845471},
}

@article{FruthSchonebornUrusov.13,
  AUTHOR   = {Fruth, Antje and Sch\"{o}neborn, Torsten and Urusov, Mikhail},
  TITLE    = {Optimal trade execution and price manipulation in order books with time-varying liquidity},
  JOURNAL  = {Math. Finance},
  FJOURNAL = {Mathematical Finance. An International Journal of Mathematics, Statistics and Financial Economics},
  VOLUME   = {24},
  NUMBER   = {4},
  PAGES    = {651--695},
  YEAR     = {2014},
  DOI      = {10.1111/mafi.12022},
  URL      = {https://doi.org/10.1111/mafi.12022},
}

@article{FruthSchonebornUrusov.19,
  AUTHOR   = {Fruth, Antje and Sch\"{o}neborn, Torsten and Urusov, Mikhail},
  TITLE    = {Optimal trade execution in order books with stochastic liquidity},
  JOURNAL  = {Math. Finance},
  FJOURNAL = {Mathematical Finance. An International Journal of Mathematics, Statistics and Financial Economics},
  VOLUME   = {29},
  NUMBER   = {2},
  PAGES    = {507--541},
  YEAR     = {2019},
  DOI      = {10.1111/mafi.12180},
  URL      = {https://doi.org/10.1111/mafi.12180},
}

@article{AckermannKruseUrusov.21disc,
  AUTHOR   = {Ackermann, Julia and Kruse, Thomas and Urusov, Mikhail},
  TITLE    = {Optimal trade execution in an order book model with stochastic liquidity parameters},
  JOURNAL  = {SIAM J. Financial Math.},
  FJOURNAL = {SIAM Journal on Financial Mathematics},
  VOLUME   = {12},
  NUMBER   = {2},
  PAGES    = {788--822},
  YEAR     = {2021},
  DOI      = {10.1137/20M135409X},
  URL      = {https://doi.org/10.1137/20M135409X},
}

@article{AckermannKruseUrusov.21cont,
  AUTHOR   = {Ackermann, Julia and Kruse, Thomas and Urusov, Mikhail},
  TITLE    = {C\`adl\`ag semimartingale strategies for optimal trade execution in stochastic order book models},
  JOURNAL  = {Finance Stoch.},
  FJOURNAL = {Finance and Stochastics},
  VOLUME   = {25},
  NUMBER   = {4},
  PAGES    = {757--810},
  YEAR     = {2021},
  DOI      = {10.1007/s00780-021-00464-5},
  URL      = {https://doi.org/10.1007/s00780-021-00464-5},
}

@article{AckermannKruseUrusov.22,
  AUTHOR   = {Ackermann, Julia and Kruse, Thomas and Urusov, Mikhail},
  TITLE    = {Reducing {O}bizhaeva--{W}ang-type trade execution problems to {LQ} stochastic control problems},
  JOURNAL  = {Finance Stoch.},
  FJOURNAL = {Finance and Stochastics},
  VOLUME   = {28},
  NUMBER   = {3},
  PAGES    = {813--863},
  YEAR     = {2024},
  DOI      = {10.1007/s00780-024-00537-1},
  URL      = {https://doi.org/10.1007/s00780-024-00537-1},
}

@article{NutzWebsterZhao.26,
    author = {Nutz, Marcel and Webster, Kevin and Zhao, Long},
    title = {{Unwinding Stochastic Order Flow: When to Warehouse Trades}},
    journal  = {Math. Finance},
    fjournal = {Mathematical Finance. An International Journal of Mathematics, Statistics and Financial Economics},
    volume = {36},
    number = {3},
    pages = {500--527},
    doi = {10.1111/mafi.70019},
    url = {https://onlinelibrary.wiley.com/doi/abs/10.1111/mafi.70019},
    year = {2026}
}

@article{ChenHorstTran.25,
    author = {Chen, Ying and Horst, Ulrich and Tran, Hoang Hai},
    title = {Optimal Trade Execution Strategy and Implementation with Deterministic Market Impact Parameters},
    journal = {Appl. Math. Finance},
    volume = {32},
    number = {1},
    pages = {1--29},
    year = {2025},
    publisher = {Routledge},
    doi = {10.1080/1350486X.2025.2537932},
    url = {https://doi.org/10.1080/1350486X.2025.2537932}
}

@article{HorstKivman.24,
    author = {Horst, Ulrich and Kivman, Evgueni},
    title = {Optimal trade execution under small market impact and portfolio liquidation with semimartingale strategies},
    journal  = {Finance Stoch.},
    fjournal = {Finance and Stochastics},
    volume = {28},
    number = {3},
    pages = {759--812},
    year = {2024},
    publisher = {Springer},
    doi = {10.1007/s00780-024-00536-2},
    url = {https://doi.org/10.1007/s00780-024-00536-2}
}

@article{CaiRosenbaumTankov.17,
    author = {Cai, Jiatu and Rosenbaum, Mathieu and Tankov, Peter},
    title = {{Asymptotic lower bounds for optimal tracking: A linear programming approach}},
    volume = {27},
    journal = {Ann. Appl. Probab.},
    fjournal = {The Annals of Applied Probability},
    number = {4},
    publisher = {Institute of Mathematical Statistics},
    pages = {2455--2514},
    year = {2017},
    doi = {10.1214/16-AAP1264},
    URL = {https://doi.org/10.1214/16-AAP1264}
}

@article{CaiRosenbaumTankov.17feedback,
    author = {Cai, Jiatu and Rosenbaum, Mathieu and Tankov, Peter},
  title   = {Asymptotic Optimal Tracking: Feedback Strategies},
  journal = {Stochastics},
  volume  = {89},
  number  = {6--7},
  pages   = {943--966},
  year    = {2017},
  doi     = {10.1080/17442508.2017.1285304}
}

@article{DolinskyiDolinsky.24,
    author = {Dolinskyi, Leonid and Dolinsky, Yan},
    title = {Optimal liquidation with high risk aversion and small linear price impact},
    volume = {47},
    journal = {Decis. Econ. Finance},
    fjournal = {Decisions in Economics and Finance},
    number = {1},
    publisher = {Springer},
    pages = {183--198},
    year = {2024},
    doi = {10.1007/s10203-024-00435-3},
    URL = {https://doi.org/10.1007/s10203-024-00435-3}
}

\end{document}